\documentclass[a4paper,fleqn]{cas-dc}
\usepackage{ulem}
\usepackage{xcolor}
\usepackage{graphicx,caption,subcaption}
\usepackage{bm}
\usepackage{subcaption}

\newtheorem{theorem}{Theorem}[section]

\newtheorem{assumption}{Assumption}[section]
\usepackage[utf8]{inputenc}
\usepackage[english]{babel}
\usepackage{comment}
\usepackage{multirow}
\usepackage[none]{hyphenat}
\usepackage{indentfirst}

\graphicspath{ {./figures/} }
\DeclareGraphicsExtensions{.png}
\usepackage{lipsum}
\usepackage[numbers]{natbib}

\usepackage{dutchcal}
\usepackage{listings}
\usepackage{scalerel}
\usepackage[ruled,longend]{algorithm2e}
\usepackage{amsmath}
\usepackage{scalerel}
\newenvironment{proof}{\paragraph{Proof:}}{\hfill$\square$}

\newcommand\reallywidehat[1]{\arraycolsep=0pt\relax%
\begin{array}{c}
\stretchto{
  \scaleto{
    \scalerel*[\widthof{\ensuremath{#1}}]{\kern-.5pt\bigwedge\kern-.5pt}
    {\rule[-\textheight/2]{1ex}{\textheight}} 
  }{\textheight} %
}{0.5ex}\\           
#1\\                 
\rule{-1ex}{0ex}
\end{array}
}
\usepackage[numbers]{natbib}

\def\tsc#1{\csdef{#1}{\textsc{\lowercase{#1}}\xspace}}
\tsc{WGM}
\tsc{QE}
\tsc{EP}
\tsc{PMS}
\tsc{BEC}
\tsc{DE}

\begin{document}
\let\WriteBookmarks\relax
\def\floatpagepagefraction{1}
\def\textpagefraction{.001}
\shorttitle{Composite Structures}
\shortauthors{S. Ma, M. Chen, and R. E. Skelton}

\title [mode = title]{Tensegrity system dynamics based on finite element method}                   



\author[1]{Shuo Ma}[
orcid=0000-0003-3789-2893]
\fnmark[1]


\address[1]{College of Civil Engineering, Zhejiang University of Technology, Hangzhou, 310014, China}

\author[2]{Muhao Chen}[type=editor,
auid=000,bioid=1,
orcid=0000-0003-1812-6835]
\cormark[1]
\fnmark[2]
\ead{muhaochen@tamu.edu}


\address[2]{Department of Aerospace Engineering, Texas A\&M University, College Station, TX, 77840, USA}

\author[2]{Robert E. Skelton}[orcid=0000-0001-6503-9115]
\fnmark[3]



\cortext[cor1]{Corresponding author. Tel.: +1 979-985-8285.}
\fntext[fn1]{Assistant Professor, College of Civil Engineering, Zhejiang University of Technology, Hangzhou, 310014, China}
\fntext[fn2]{Postdoctoral Researcher, Department of Aerospace Engineering, Texas A\&M University, College Station, TX, USA.}
\fntext[fn3]{TEES Eminent Professor, Department of Aerospace Engineering, Texas A\&M University, College Station, TX, USA.}


\begin{abstract}
This study presents a finite element analysis approach to non-linear and linearized tensegrity dynamics based on the Lagrangian method with nodal coordinate vectors as the generalized coordinates. In this paper, nonlinear tensegrity dynamics with and without constraints are first derived. The equilibrium equations in three standard forms (in terms of nodal coordinate, force density, and force vectors) and the compatibility equation are also given. Then, we present the linearized dynamics and modal analysis equations with and without constraints. The developed approach is capable of conducting the following comprehensive dynamics studies for any tensegrity structures accurately: 1. Performing rigid body dynamics with acceptable errors, which is achieved by setting relatively high stiffness for bars in the simulation. 2. Simulating FEM dynamics accurately, where bars and strings can have elastic or plastic deformations. 3. Dealing with various kinds of boundary conditions, for example, fixing or applying static/dynamic loads at any nodes in any direction (i.e., gravitational force, some specified forces, or arbitrary seismic vibrations). 4. Conducting accurate modal analysis, including natural frequency and corresponding modes. Three examples, a double pendulum, a cantilever truss with external force, and a double prism tensegrity tower, are carefully selected and studied. The results are compared with rigid body dynamics and FEM software ANSYS. This study provides a deep insight into structures, materials, performances, as well as an interface towards integrating control theories.

\end{abstract}




\maketitle

\section{Introduction}
Tensegrity is a coined word: tension and integrity, by Buckminister Fuller \cite{fuller1982synergetics} for the art form created by Ioganson (1921) and Snelson (1948) \cite{lalvani1996origins}. The tensegrity system is a stable network of compressive members (bars/struts) and tensile members (strings/cables). By definition, it is clear that the most fundamental property of the tensegrity system is that all the one-dimensional structural members are axially loaded \cite{ma2020design}. That is to say, the overall structure can be designed along the load path to make the best of each structure member since bars are best in taking compression, and strings are best in taking tension. In fact, a few research on form-finding \cite{koohestani2017analytical,yuan2017form,Zhang2006Adaptive} and topology optimization \cite{lee2016novel,ma2019new,xu2018improved} have shown how to find the optimal load path and where to locate structure members with given design objectives.

Thus, the advantages of the tensegrity structure are straight forward: 1. The mass of the structure to take given loads can be greatly reduced. In fact, many structures have been restudied and redesigned by the tensegrity paradigm to achieve a lightweight objective. For example, Skelton \textit{et al.} proved structure mass to take a compressive load can be greatly reduced by T-Bar and D-bar structures \cite{skelton2009tensegrity}. Chen and Skelton presented a general approach to minimal mass tensegrity considering local and global failure,  solid and hollow bar, gravity, and stiffness calculations \cite{chen2020general}. 2. There is no material bending. Thus, the uncertainty of a structure member is only along its length, which provides a more actuate model \cite{chen2020habitat}. 3. One can change the shape of the structure easily by tuning the length of the strings \cite{fraddosio2017morphology,guest2011stiffness}. 4. Since the strings can be prestressed, the stiffness of the structure can be tuned by the tensions in the strings \cite{liu2019unraveling,yildiz2019effective,zhang2018tensegrity}. 5. The soft structure can be used to absorb energy \cite{ma2018meta,miranda2020mechanics,pajunen2019design,yang2019deployment}. 6. The tensegrity paradigm also promotes the integration of structure and control design, since bars and strings can also be actuators and sensors \cite{chen2020design,kim2020rolling,rieffel2009automated,zhang2018automatically}. 

Tensegrity has shown its great attraction to both artists and engineers, a few research on tensegrity dynamics has been conducted. The existing tensegrity dynamics can be classified into two categories based on the assumptions of whether the bars are rigid or not \cite{kan2018nonlinear}. The first category belongs to rigid body dynamics derived by Newton-Euler's principle or analytical dynamics with assumptions that bars are rigid and strings are linear elastic. For example, Sultan \textit{et al.} derived linearized equations of motion for tensegrity models around arbitrary equilibrium configurations \cite{sultan2002linear}. Skelton presented one of the simplest dynamics forms for class-1  structure by using non-minimal coordinates and assuming the compressive elements to have no inertia about the longitudinal axis \cite{skelton2005dynamics}. Kan \textit{et al.} presented a sliding cable element for multibody dynamics with an application to the deployment of clustered tensegrity \cite{Kan2017A}. Cefalo and Mirats-Tur proposed a dynamic model based on the Lagrangian method for class-1 tensegrity systems with quaternions as the variables
\cite{cefalo2011comprehensive}. Goyal \textit{et al.} presented a compact matrix form of tensegrity dynamics by including massive strings \cite{goyal2019tensegrity}, a corresponding general software for modeling of any tensegrity structures can be found in \cite{goyal2019motes}. Recently, Goyal \textit{et al.} extended the model of the nonlinear dynamics to modulate the torque produced by the network of spatially distributed gyroscopes \cite{goyal2020gyroscopic}. The second one is non-rigid body dynamics formulated by the FEM by assuming that all structure members are elastic/plastic. For example, Murakami studied the static and dynamic equations of tensegrity with large deformation in Eulerian and Lagrangian formulations \cite{murakami2001static}. Faroughi \textit{et al.} presented a non-linear dynamic analysis of space truss structures based on the dynamics of 3D co-rotational (CR) rods \cite{faroughi2015non}. Rimoli developed a physics-based reduced-order model to capture the buckling and post-buckling behavior of bars \cite{rimoli2018reduced}. Kan \textit{et al.} derived the dynamic analysis of clustered tensegrity structures via the framework of the positional formulation FEM \cite{kan2018nonlinear}. However, most of these dynamics equations are achieved by deriving the dynamics of one element and stacking all the structural elements into an assembled matrix or vector form. For the insight knowledge of the nonlinear tensegrity dynamics and future convenience for the field of structural control, a closed-form of dynamics derived from a system-level is needed, which is presented in this paper. 





This paper is organized as follows: Section \ref{section 2} describes bar and string assumptions, nodal coordinates and connectivity matrices notations, and geometric and physical properties of the tensegrity system in compact vector forms. Section \ref{Seciton 3} formulates the shape function of an element, kinetic energy, strain, and gravitational potential energy of the whole structure. Then, tensegrity dynamics with and without boundary constraints are derived by the Lagrangian method. By neglecting the time derivative terms in the dynamics equation, Section \ref{Section 4} gives the equilibrium equations in three standard equivalent forms (in terms of nodal vector, force density, and force vector) and the compatibility equation. Section \ref{section 5} derives the linearized tensegrity dynamics and modal analysis equations with and without boundary constraints. Section \ref{section 6} demonstrates three examples (dynamic response of a double pendulum, dynamics response and modal analysis of a cantilever truss with an external force, and seismic analysis of a tensegrity tower) and compares results with rigid body dynamics and FEM software ANSYS. Section \ref{section 7} summarises the conclusions.

\section{Notations of the tensegrity system}
\label{section 2}
\subsection{Assumptions of structural members}

Under the following assumptions of structural members (bars and strings), the mathematical formulation of any tensegrity systems is established.
\begin{assumption}
The structural members (bars and strings) in the tensegrity system have these properties: \\
1). The structural members are axially loaded, all structural members are connected by frictionless pin-joints. \\
2). The structural members are not rigid, and they are allowed to have elastic or plastic deformation. \\
3). The structural members have negligible inertia about their longitudinal axes.\\
4). Each structural member is homogeneous along its length and of an equal cross-section. Thus, the mass of each structural member is distributed uniformly along its length.\\
5). If $||\bm{s}_{i0}|| > ||\bm{s}_i||$, where the rest length and actual length of the $i^{th}$ string are denoted by $||\bm{s}_{i0}||$ and $||\bm{s}_i||$, and $\lVert\bm{v}\rVert$ is the Euclidean norm of vector $\bm{v}$, since a string can never push along its length, tension in the string should be substituted to zero.\\
\end{assumption}

\subsection{Nodal coordinates}

The position of each node in the structure can be expressed in any frame, we choose to label them with Cartesian coordinates in an inertially fixed frame. Assume the tensegrity structure has $n_n$ number of nodes, the X-, Y-, and Z-coordinates of the \textit{i}th node $\bm{n}_i$ ($i = 1, 2, \cdots, n_n$) can be labeled as $x_i$, $y_i$, and $z_i$. One can also write $\bm{n}_i \in \mathbb{R}^3$ in a vector form:
\begin{align}
\label{ni}
    \bm{n}_i=\begin{bmatrix}  x_i & y_i & z_i\end{bmatrix} ^T.
\end{align}
By stacking $\bm{n}_i$ for $i = 1, 2, \cdots, n_n$ together, we can obtain the nodal coordinate vector $\bm{n} \in \mathbb{R}^{3 n_n}$ for the whole structure:
\begin{align}
     \bm{n}= \begin{bmatrix} \bm{n}_1^T & \bm{n}_2^T & \cdots & \bm{n}_{n_n}^T\end{bmatrix}^T,
\end{align}
or in a matrix form, which is called nodal coordinate matrix $\bm{N}\in \mathbb{R}^{3\times n_n}$:
\begin{align}
   \bm{N} = \begin{bmatrix}
    \bm{n}_1 & \bm{n}_2 & \cdots & \bm{n}_{n_n}
    \end{bmatrix}.
\end{align}









\subsection{Connectivity matrix}
Connectivity matrices denote the topology of the structure or, in other words, how the structural members (bars and strings) are connected at each node. Conventionally, the connectivity matrices contain two types: string connectivity and bar connectivity, labeled as $\bm{C}_s\in \mathbb{R}^{\alpha \times n_n}$ and $\bm{C}_b \in \mathbb{R}^{\beta \times n_n}$, where $\alpha$ and $\beta$ are the number of strings and bars in the structure \cite{goyal2019motes}. 

Since both bars and strings are allowed to have elastic or plastic deformation, we do not need to distinguish the connectivity by the types of structural members in this FEM formulation. Thus, we use a matrix  $\bm{C} \in \mathbb{R}^{n_e \times n_n}$ to represent the topology of the whole structure, where $n_e$ is the number of all the structural elements, which satisfies $n_e = \alpha + \beta$. The \textit{i}th row of $\bm{C}$, denoted as $\bm{C}_i = [\bm{C}]_{(i,:)} \in \mathbb{R}^{1\times n_n}$, represents the \textit{i}th structural element, starting form node \textit{j} ($\textit{j} = 1, 2, \cdots, n_n$) to node \textit{k} ($\textit{k} = 1, 2, \cdots, n_n$), shown in Fig.\ref{geometry}. The \textit{m}th ($\textit{m}= 1, 2, \cdots, n_n$) entry of $\bm{C}_{i}$ satisfies:
\begin{align}
[\bm{C}]_{im}=\left\{
\begin{aligned}
-1 &,~ m=j\\
1 &,~ m=k\\
0 &,~ m=else 
\end{aligned}
\right..
\end{align}
For $n_e$ number of structural elements, the overall structure connectivity matrix $\bm{C}\in  \mathbb{R}^{n_e \times n_n}$ can be written as:
\begin{align}
\bm{C} = \begin{bmatrix} \bm{C}_1^T & \bm{C}_2^T & \cdots & \bm{C}_{n_e}^T\end{bmatrix} ^T .
\end{align}

Define the nodal coordinate vector of the \textit{i}th element $\bm{n}_i^e  \in \mathbb{R}^6$ as:
\begin{align}
\bm{n}_i^e=\begin{bmatrix}
\bm{n}_j \\
\bm{n}_k \\
\end{bmatrix}=\begin{bmatrix} x_j & y_j& z_j&x_k&y_k&z_k\end{bmatrix}^T.
\end{align}
One can also abstract $\bm{n}_i^e$ from the structure nodal coordinate vector $\bm{n}$:
\begin{align}
    \bm{n}_i^e=\Bar{\bm{C}}_i\otimes \textbf{I}_3\bm{n},
    \label{n_i^e}
\end{align}
where $\textbf{I}_{3} \in \mathbb{R}^{3\times 3}$ is a identity matrix, $\Bar{\bm{C}}_i$ is a self-defined transformation matrix, whose \textit{p}th column satisfies:
\begin{align}
[\bar{C}_{i}]_{(:,p)}=\left\{\begin{array}{ll}
{\begin{bmatrix}1&0\end{bmatrix}^T}, & p=j \\
{\begin{bmatrix}0&1\end{bmatrix}^T}, & p=k \\
{\begin{bmatrix}0&0\end{bmatrix}^T}, & p=else
\end{array}\right..
\end{align}



\subsection{Geometric properties of the structural elements}

\begin{figure}
    \centering
    \includegraphics[scale=0.8]{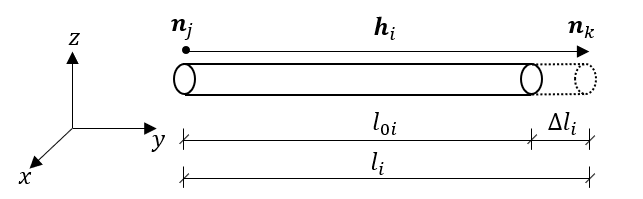}
    \caption{Structure member vector $\bm{h}_i$, determined by node $\bm{n}_j$ and node $\bm{n}_k$ in the Cartesian coordinates, has a length of $l_i  = ||\bm{h}_i||= l_{0i} + \Delta l_i$, where $l_{0i}$ is the rest length and $\Delta l_i$ is the displacement.}
    \label{geometry}
\end{figure}



Let us look at the $i$th structure element $\bm{h}_i$, its geometry properties is shown in Fig.\ref{geometry}, the element vector is given by:
\begin{align}
    \bm{h}_i=\bm{n}_k-\bm{n}_j=\bm{C}_i\otimes\textbf{I}_3\bm{n}.
\end{align}
Stack all the structure elements in a matrix form, one can obtain:
\begin{align}
    \bm{H}=\begin{bmatrix}\bm{h}_1 & \bm{h}_2&\cdots&\bm{h}_{n_e}\end{bmatrix}=\bm{N}\bm{C}^T.
\end{align}
The length of the \textit{i}th structure element $l_i$ satisfies:
\begin{align}
    l_i=\lVert \bm{h}_i\rVert=(\bm{n}^T(\bm{C}_i^T\bm{C}_i)\otimes\textbf{I}_3\bm{n})^\frac{1}{2}.
    \label{li}
\end{align}
Then, the overall structure element length vector $\bm{l} \in \mathbb{R}^{n_e}$ is:
\begin{align}
    \bm{l}=\begin{bmatrix}l_1&l_2&\cdots&l_{n_e}\end{bmatrix}^T.
\end{align}
The rest length vector $\bm{l}_0 \in \mathbb{R}^{n_e}$  of the whole structure is:
\begin{align}
      \bm{l}_0=\begin{bmatrix} l_{01}&l_{02}&\cdots&l_{0{n_e}}\end{bmatrix}^T, 
\end{align}
where rest length is defined as the length of an structure element with no tension or compression. 

\subsection{Physical properties of the structural elements}

\begin{figure}
    \centering
    \includegraphics[width=2.5in]{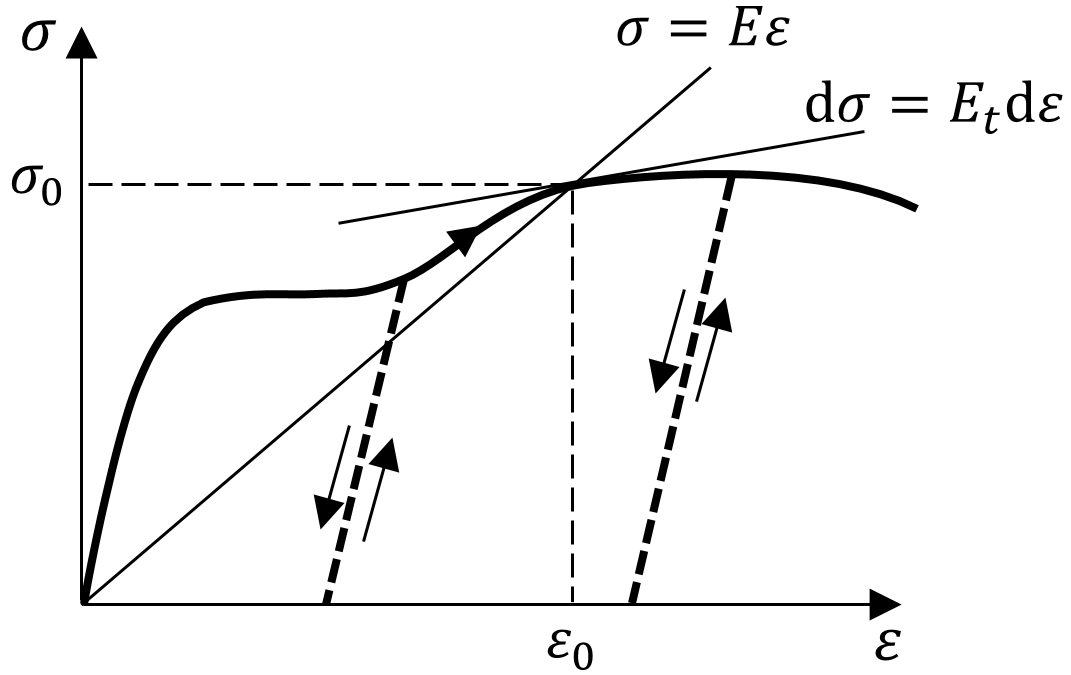}
    \caption{A typical stress-strain curve of structure elements, where $E$ and $E_t$ are called secant modulus and tangent modulus of the material. The curve includes elastic or plastic deformation phases. The dotted lines indicate stress-strain levels for unloading cases at certain points of the stress-strain curve.}
    \label{stress strain curve}
\end{figure}

A typical stress-strain curve of an element is given in Fig.\ref{stress strain curve}. The stress in the element satisfies:
\begin{equation}
    \sigma=E\epsilon,
    \label{sigma}
\end{equation}
where $E$ is the secant modulus and $\epsilon$ is the strain, and this equation can represent the stress of any material including linear elastic, multi-linear elastic, plastic, etc. The derivative of Eq. (\ref{sigma}) is: 
\begin{equation}
       \mathrm{d} \sigma=E_t\mathrm{d} \epsilon,
    \label{d sigma}
\end{equation}
where $E_t$ is tangent modulus. For elastic material, the secant modulus is identical to its tangent modulus. We discuss the elastic and plastic properties of materials here because later we will show that the developed dynamics are capable of doing analysis of both kinds of materials.

Suppose material density is $\rho$, the cross section area, secant modulus, tangent modulus of the \textit{i}th element are respectively $A_i$, $E_i$  $E_{ti}$, the element mass $m_i$ satisfies $m_i =\rho A_i l_{0i}$. Denote the cross section area vector, mass vector, secant modulus and tangent modulus vector of the structure as $\bm{A}$, $\bm{m}$, $\bm{E}$ and $\bm{E}_t$ $\in \mathbb{R}^{n_e}$, one can write:
\begin{align}
    \bm{A} & =\begin{bmatrix}A_{1}& A_{2}&\cdots&A_{n_e}\end{bmatrix}^T,\\
    \bm{m} & =\begin{bmatrix}m_{1}&m_{2}&\cdots& m_{n_e}\end{bmatrix}^T=\rho\hat{\bm{A}}\bm{l}_0 ,\\
   \bm{E} & =\begin{bmatrix}E_{1}&E_{2}&\cdots&E_{n_e} \end{bmatrix}^T,\\
   \bm{E}_t &=\begin{bmatrix}E_{t1}&E_{t2}&\cdots&E_{tn_e}
   \end{bmatrix}^T,
\end{align}
where $\hat{\bm{v}}$ transforms vector $\bm{v}$ into a diagonal matrix, whose diagonal entries are the elements of vector $\bm{v}$ and elsewhere are zeros. 

The internal force of the \textit{i}th element is $t_i =A_i\sigma_i= E_iA_i(l_i-l_{0i})/l_{0i}$, the internal force vector of the structure $\bm{t} \in \mathbb{R}^{n_e}$ can be written as:
\begin{align}
    \bm{t}=
    \begin{bmatrix} t_1 & t_2 & \cdots&t_{n_e} \end{bmatrix}^T= 
    \hat{\bm{E}}\hat{\bm{A}}\hat{\bm{l}}_0^{-1}(\bm{l}-\bm{l}_0).
    \label{t force vector}
\end{align}
Force density of the \textit{i}th element is given by $x_i=t_i/l_i$, the force density vector of all the structure elements is:
\begin{align}
\bm{x}=\hat{\bm{l}}^{-1} \bm{t}=\hat{\bm{E}} \hat{\bm{A}}(\bm{l}_{0}^{-1}-\bm{l}^{-1}),
\label{force density}
\end{align}
where $\bm{\bm{v}}^{-1}$ represents a vector whose entry is the reciprocal of its corresponding entry in $\bm{v}$. 
The force density vector $\bm{x}$ is normally defined in the from of $\bm{x} = \begin{bmatrix} \bm{\lambda}^T &\bm{ {\gamma}}^T\end{bmatrix}^T$ with the information of ${\bm{\lambda}}$ and ${\bm{\gamma}}$ are force densities in the bars and strings \cite{chen2020general}. We should point out that Eqs. (\ref{t force vector}) and (\ref{force density}) can be used to compute force vector and force density vector for either elastic or plastic materials by using different secant modulus $\textbf{E}$ of the materials.


\section{Nonlinear tensegrity dynamics formulation}
\label{Seciton 3}
\subsection{Energy equation formulation}

\subsubsection{Shape function of the structure element}


\begin{figure}
    \centering
    \includegraphics[scale=0.35]{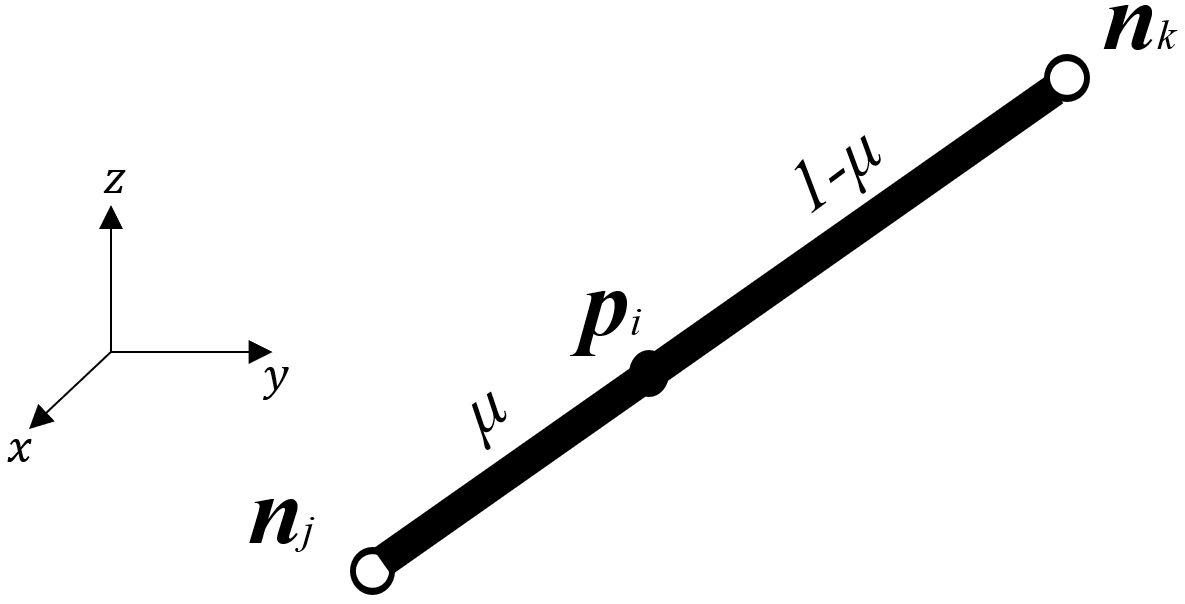}
    \caption{Shape function of an element, scalar $\mu$ helps to locate the position of point $\bm{p}_i$ on the $i$th structure element in between node $\bm{n}_j$ and node $\bm{n}_k$ in the Cartesian coordinates.}
    \label{shape function}
\end{figure}







Since the structure member is axially loaded, the displacement of the material particles are along the bar/string vectors. We assume the displacements of material particles on the structure member are in a uniform manner \cite{bathe2007finite}. Here, we introduce a scalar $\mu$ to help expressing the coordinates of point $\bm{p}_i$ on the $i$th member between node $\bm{n}_j$ and node $\bm{n}_k$ in the \textit{i}th structure element, shown in Fig.\ref{shape function}. Thus, the location of a point $\bm{p}_i$ on the structure member can be computed as a linear function in terms of $\mu$:
\begin{align}
    \bm{p_i}=\begin{bmatrix} 1&\mu \end{bmatrix} \otimes \textbf{I}_3 \begin{bmatrix}\bm{a}_0 \\\bm{a}_1\end{bmatrix}
    \label{shape function r},
\end{align}
where $\bm{a}_0$ and $\bm{a}_1$  $\in \mathbb{R}^3$ are unknowns. Substitute the nodal coordinate of $\bm{n}_j$, $\bm{n}_k$ with $\mu$ = 0 and 1 into Eq. (\ref{shape function r}), we have:
\begin{align}\label{a_0_1}
\begin{bmatrix}1& 0\\ 1& 1\end{bmatrix}\otimes \textbf{I}_3 \begin{bmatrix}\bm{a}_0 \\\bm{a}_1\end{bmatrix} =\begin{bmatrix}\bm{n}_j \\\bm{n}_k\end{bmatrix}.
\end{align}
Then, the solution of Eq. (\ref{a_0_1}) is:
\begin{align}
\begin{bmatrix}\bm{a}_0 \\\bm{a}_1\end{bmatrix}=\begin{bmatrix}1& 0\\ -1& 1\end{bmatrix}\otimes\textbf{I}_3  \begin{bmatrix}\bm{n}_j \\\bm{n}_k\end{bmatrix}
\label{shape function a0a1}.
\end{align}
Substitute Eq. (\ref{shape function a0a1}) into Eq. (\ref{shape function r}), we have: 
\begin{align}    
    \bm{p}_i&=\begin{bmatrix} 1-\mu & \mu\end{bmatrix}\otimes\textbf{I}_3\begin{bmatrix}\bm{n}_j \\\bm{n}_k\end{bmatrix}=\bm{N}^e \bm{n}_i^e, \label{r=N^en_i^e} \\
        \bm{N}^e&=\begin{bmatrix}1-\mu&\mu\end{bmatrix}\otimes \textbf{I}_3,
\end{align}
where $\bm{N}^e\in \mathbb{R}^{3\times 6}$ is usually called the shape function of a structure element.
\subsubsection{Kinetic energy}

The kinetic energy $T$ for the whole structure equals the sum of kinetic energy of material particles in all the structure elements, which can be written as a function of particle velocity $\dot{\bm{p}}_i$:
\begin{align}
    \label{T1}
    T=\sum\limits_{i=1}^{n_e}\frac 12\int\limits_{0}^1m_i\lVert\dot{\bm{p}}_i\rVert^2d\mu. 
\end{align}
Substitute Eq. (\ref{n_i^e}) and Eq. (\ref{r=N^en_i^e}) into Eq. (\ref{T1}), we have:
\begin{align}
    T &= \sum\limits_{i=1}^{n_e}\frac 12\int\limits_{0}^1m_i(\bm{N}^e\bar{\bm{C}}_i\otimes \textbf{I}_3\dot{\bm{n}})^2d\mu\\
    &=\sum\limits_{i=1}^{n_e}\frac {m_i}{12}\dot{\bm{n}}^T(\bar{\bm{C}}_i^T\begin{bmatrix}2&1\\1&2\end{bmatrix}  \bar{\bm{C}}_i)\otimes \textbf{I}_3\dot{\bm{n}} \\
    &=\sum\limits_{i=1}^{n_e}\frac {m_i}{12}\dot{\bm{n}}^T(\bar{\bm{C}}_i^T(\begin{bmatrix}1\\1\end{bmatrix}[1~1]+ \lfloor\begin{bmatrix}1\\1\end{bmatrix}[1~1]\rfloor )\bar{\bm{C}}_i)\otimes \textbf{I}_3\dot{\bm{n}} \\
    &=\sum\limits_{i=1}^{n_e}\frac {1}{12}\dot{\bm{n}}^T(|\bm{C}|_i^Tm_i|\bm{C}|_i+ \lfloor |\bm{C}|_i^Tm_i |\bm{C}|_i\rfloor)\otimes \textbf{I}_3 \dot{\bm{n}} \\   
   &=\frac {1} {12}\dot{\bm{n}}^T(|\bm{C}|^T\hat{m}|\bm{C}|+ \lfloor|\bm{C}|^T\hat{m}|\bm{C}|\rfloor)\otimes \textbf{I}_3\dot{\bm{n}}\\
   &=\frac {1} {2}\dot{\bm{n}}^T \bm{M}\dot{\bm{n}},
\end{align}
where $|\bm{V}|$ is an operator getting the absolute value of each element for a given matrix, and the operator $\lfloor \bm{V} \rfloor$ sets every off-diagonal element of the square
matrix to zero. $\bm{M}\in \mathbb{R}^{3n_n\times 3n_n}$ is called the mass matrix of the structure:
\begin{align}
   \bm{M} =\frac {1} {6}(|\bm{C}|^T\hat{\bm{m}}|\bm{C}|+ \lfloor|\bm{C}|^T\hat{\bm{m}}|\bm{C}|\rfloor)\otimes \textbf{I}_3.
    \label{M}
\end{align}

Since the matrix $\bm{M}$ is symmetric, we can have the following equation:
\begin{align}
    \frac{\mathrm{d}}{\mathrm{d}t}\frac{\partial T}{\partial \dot{\bm{n}}}=\bm{M}\ddot{\bm{n}}.
    \label{Mddn}
\end{align}
Note that we use denominator layout notation in matrix calculus, which means the derivative of a scalar by a column vector is still a column vector.
\subsubsection{Strain potential energy}
We consider elastic and plastic deformation of structure members, to unify the two cases, the strain potential energy $V_e$ of the whole structure caused by elements' internal force can be written into an integral form:
\begin{align}
  \label{Ve}
V_e&=\sum\limits_i^{n_e}V_{ei}\\
&=\sum\limits_i^{n_e}\int_{l_{0i}}^{l_i}t_i\mathrm{d}u\\
&=\sum\limits_i^{n_e}\int_{l_{0i}}^{l_i}\frac{E_iA_i(u-l_{0i})}{l_{0i}}\mathrm{d}u,
\end{align}
where $\mathrm{d}u$ is the differential of the structure member length. The derivative of strain potential energy $V_e$ with respect to nodal coordinate vector $\bm{n}$ is:
\begin{align}
    \frac{\partial V_e}{\partial\bm{n}}
    &=\sum\limits_i^{n_e}\frac{\partial V_{ei}}{\partial l_i}\frac{\partial l_i}{\partial \bm{n}}\\
    &=\sum\limits_i^{n_e}\frac{E_iA_i(l_i-l_{0i})}{l_{0i}}\frac{\partial l_i}{\partial \bm{n}}\\
&=\sum\limits_i^{n_e}t_i\frac{\partial l_i}{\partial \bm{n}}
    \label{dV_edn1}.
\end{align}
The derivative of element's length $l_i$ with respect to nodal coordinate vector $\bm{n}$ can be obtained from Eq. (\ref{li}):
\begin{align}
    \frac{\partial l_i}{\partial \bm{n}}=\frac{(\bm{C}_i^T\bm{C}_i)\otimes\textbf{I}_3\bm{n}}{l_i}
    \label{dli_dni}.
\end{align}
Substitute Eq. (\ref{dli_dni}) into Eq. (\ref{dV_edn1}), and use the definition of force density $x_i = f_i/l_i$ in the \textit{i}th structure element, we have:
\begin{align}
    \frac{\partial V_e}{\partial\bm{n}}&=\sum\limits_i^{n_e}x_i(\bm{C}_i^T\bm{C}_i)\otimes\textbf{I}_3\bm{n}\\
    &=(\bm{C}^T\hat{\bm{x}}\bm{C})\otimes\textbf{I}_3\bm{n}\\
    &=\bm{Kn},
    \label{Kn}
\end{align}
where $\bm{K} \in \mathbb{R}^{3n_n \times 3n_n}$ is the stiffness matrix of the tensegrity structure:
\begin{align}
    \bm{K}=(\bm{C}^T\hat{\bm{x}}\bm{C})\otimes\textbf{I}_3.
    \label{K}
\end{align}

\subsubsection{Gravitational potential energy}
For many cases that the tensegrity structures are in the presence of gravity field. Suppose the gravity force is exerted towards the negative direction of the Z-axis, the acceleration of gravity is $g$, for example, on earth $g=9.8 m/s^2$. The gravitational potential energy $V_g$ can be written as:
\begin{align} 
    V_g &=\sum\limits_i^{n_e}\frac{m_ig}{2}(z^i_j+z^i_k)\\ 
    &=\sum\limits_i^{n_e}\frac{m_ig}{2}|\bm{C}_i|\otimes
    \begin{bmatrix}
    0 & 0 & 1
    \end{bmatrix}
\bm{n}\\
    &=\frac{g}{2}\bm{m}^T|\bm{C}|\otimes    \begin{bmatrix}
    0 & 0 & 1
    \end{bmatrix}\bm{n},
\end{align}
where $z^i_j$ and $z^i_k$ are Z-coordinates of node $\bm{n}_j$ and node $\bm{n}_k$ of the \textit{i}th structure element. $(z^i_j+z^i_k)/2$ is the Z-coordinate of mass center of the \textit{i}th element. Then, the partial derivative of $V_g$ with respect to $\bm{n}$ is:
\begin{align}
    \frac{\partial V_g}{\partial \bm{n}}=\frac{g}{2}(|\bm{C}|^T\bm{m})\otimes    \begin{bmatrix}
    0 & 0 & 1
    \end{bmatrix}^T=\bm{g},
    \label{g}
\end{align}
where $\bm{g} \in \mathbb{R}^{3 n_n}$ is the gravitational force vector in all nodes. For structure analysis without gravity, one can just set $\bm{g} = \bm{0}$. 

\subsection{Tensegrity dynamics formulation based on Lagrangian method}



\begin{theorem}
The finite element formulation for nonlinear tensegrity dynamics in the presence of gravity is given by:
\begin{equation}
\bm{M} \ddot{\bm{n}}+\bm{D} \dot{\bm{n}}+\bm{K} \bm{n}=\bm{f}_{e x}-\bm{g},
\label{dynamics2}
\end{equation}
where $\bm{M}$, $\bm{D}$, and $\bm{K}$ are mass, damping, and stiffness matrices given in Eqs. (\ref{M}), (\ref{f_d}), and (\ref{K}), $f_{ex}$ is external forces on the structure nodes, and $\bm{g}$ is gravity vector show in Eq. (\ref{g}).
\end{theorem}

\begin{proof}

The dynamics of tensegrity structures can be derived by Lagrangian method. The Lagrangian equation is:
\begin{align}
    \frac{\mathrm{d}}{\mathrm{d}t}\frac{\partial L}{\partial \dot{\bm{q}}}-\frac{\partial L}{\partial \bm{q}}=\bm{f}_{np},
    \label{Lagarangian's equation}
\end{align}
where $L=T-V$ is the Lagrangian function, and $T$ and $V$ are the kinetic energy and potential energy of the system, $\bm{f}_{np}$ is the non-potential force vector on the nodes of the tensegrity structures, and $\bm{q}$ is the generalized coordinate of the system, which is nodal coordinate vector $\bm{n}$ in our derivation. The potential energy of the whole structure is the sum of strain energy $V_e$ and gravitational potential energy $V_g$:
\begin{align}
    V=V_e+V_g.
\end{align}
Then, the Lagrangian's function can be written as:
\begin{align}
    L=T-(V_e+V_g).
    \label{L}
\end{align}
Substitute Eqs. (\ref{Mddn}), (\ref{Kn}), (\ref{g}), (\ref{L}) into Eq. (\ref{Lagarangian's equation}), we have the dynamic equation of the tensegrity structure:
\begin{align}
\bm{M}\ddot{\bm{n}}+\bm{Kn}=\bm{f}_{np}-\bm{g},
\label{dynamics1}
\end{align}
where $\bm{M}\in \mathbb{R}^{3n_n\times 3n_n}$, $\bm{K} \in \mathbb{R}^{3n_n\times 3n_n}$, $\bm{g}\in \mathbb{R}^{3n_n}$ are mass matrix, stiffness matrix and gravitational force vector, given in Eq. (\ref{M}), Eq. (\ref{K}), and Eq. (\ref{g}). The non-potential force $\bm{f}_{np}$ is the sum of damping force $\bm{f}_d$ and external force $\bm{f}_{ex}$:
\begin{align}
    \bm{f}_{np}=\bm{f}_{d}+\bm{f}_{ex}.
    \label{f_np}
\end{align}
The damping force is assumed to be linear in terms of $\dot{n}$ as:
\begin{align}
    \bm{f}_{d}=-\bm{D}\dot{\bm{n}},
    \label{f_d}
\end{align}
where $\bm{D}\in \mathbb{R}^{3n_n\times 3n_n}$ is the damping matrix of the structure. Substitute Eq. (\ref{f_np}) and Eq. (\ref{f_d}) into Eq. (\ref{dynamics1}), we can obtain a standard form of the dynamics equation as in Eq. (\ref{dynamics2}).

\end{proof}


\subsection{Tensegrity dynamics with constraints}
By constraints, we mean, for many cases, the position, velocity, or acceleration of some nodes in the structure are fixed/given. Adding these constraints to the dynamics will restrict the motion in certain dimensions, thus the dynamics Eq. (\ref{dynamics2}) can be reduced into a smaller space. Thus, to compute the reduced-order dynamics, free nodes and fix nodes must be separated. 




Let us define vector $\bm{a}=\begin{bmatrix}a_{1} & a_{2}& \cdots & a_{n_{a}}\end{bmatrix}^{T} \in \mathbb{R}^{n_{a}}$ and vector $\bm{b}=\begin{bmatrix}b_{1}& b_{2}& \cdots & b_{n_{b}}\end{bmatrix}^{T} \in \mathbb{R}^{n_{b}}$, in which the element values of $\bm{a}$ and $\bm{b}$ are the indices of free and constrained entries in the nodal coordinate vector $\textbf{n}$. $n_a$ and $n_b$ is the number of free and constrained nodal coordinates, and they satisfy $n_a+n_b=3n_n$.  We use $\bm{n}_a$ and $\bm{n}_b$ to represent the free and constrained nodal coordinate vector. $\bm{E}_{a} \in \mathbb{R}^{3 n_{n} \times n_{a}}$ and $\bm{E}_{b} \in \mathbb{R}^{3 n_{n} \times n_{b}}$ are the matrix to abstract $\bm{n}_a$ and $\bm{n}_b$ from $\bm{n}$:
\begin{align}
\bm{E}_{a}(:, i)=\textbf{I}_{3 n}\left(:, a_{i}\right), ~\bm{E}_{b}(:, i)=\textbf{I}_{3 n}\left(:, b_{i}\right).
\end{align}
The relation between $\bm{n}_a$, $\bm{n}_b$, and $\bm{n}$ is:
\begin{align}
\bm{n}_{a}=\bm{E}_{a}^{T} \bm{n},~ \bm{n}_{b}=\bm{E}_{b}^{T} \bm{n}.
\end{align}
Note that $\begin{bmatrix}\bm{E}_{a} &  \bm{E}_{b}\end{bmatrix}$ is an orthonormal matrix, so given $\bm{n}_a$ and $\bm{n}_b$, the nodal coordinate vector $\bm{n}$ can be obtained by:
\begin{align}\bm{n}=
\begin{bmatrix}
\bm{E}_{a}^{T} \\
\bm{E}_{b}^{T}
\end{bmatrix}^{-1}
\begin{bmatrix}\bm{n}_{a} \\
\bm{n}_{b}
\end{bmatrix}=
\begin{bmatrix}
\bm{E}_{a} & \bm{E}_{b}\end{bmatrix}
\begin{bmatrix}
\bm{n}_{a} \\
\bm{n}_{b}
\end{bmatrix}.
\label{n=[Ea Eb][na;nb]}
\end{align}
\begin{theorem}
The finite element formulation for nonlinear tensegrity dynamics in the presence of constraints and gravity is given by:
\begin{align}\nonumber
\bm{M}_{a a} \ddot{\bm{n}}_{a}+\bm{D}_{a a} \dot{\bm{n}}_{a}+\bm{K}_{a a} \bm{n}_{a}= & \bm{E}_{a}^{T} \bm{f}_{e x} -\bm{M}_{a b} \ddot{\bm{n}}_{b}-\bm{D}_{a b} \dot{\bm{n}}_{b} \\ & -\bm{K}_{a b} \bm{n}_{b}-\bm{E}_{a}^{T} \bm{g},
\label{dynamics equation reduced order}
\end{align}
where $\bm{M}_{a a}$ and $\bm{M}_{a b}$ are mass matrices. $\bm{D}_{a a}$ and $\bm{D}_{a b}$ are damping matrices. $\bm{K}_{a a}$ and $\bm{K}_{a b}$ are stiffness matrices. $\bm{f}_{ex}$ is external forces on the structure nodes, and $\bm{g}$ is gravity vector, which satisfy:
\begin{align}\label{aa_ab1}
    \bm{M}_{aa} & =\bm{E}_{a}^{T} \bm{M} \bm{E}_{a},~
\bm{M}_{ab}=\bm{E}_{a}^{T} \bm{M} \bm{E}_{b}, \\ \label{aa_ab2}
\bm{D}_{aa} & =\bm{E}_{a}^{T} \bm{D} \bm{E}_{a},~
\bm{D}_{ab}=\bm{E}_{a}^{T} \bm{D} \bm{E}_{b}, \\ \label{aa_ab3}
\bm{K}_{aa} & =\bm{E}_{a}^{T} \bm{K} \bm{E}_{a},~
\bm{K}_{ab}=\bm{E}_{a}^{T} \bm{K} \bm{E}_{b},
\end{align}
and $\bm{M}$, $\bm{D}$, $\bm{K}$, and $\bm{g}$ are given in Eqs. (\ref{M}), (\ref{f_d}), (\ref{K}), and (\ref{g}).
\end{theorem}

\begin{proof}
If the tensegrity structure has boundary constraints, the degree of freedom reduces to $n_a$, thus the free nodal coordinate vector $\bm{n}_a$ is the generalized coordinate. Then, the Lagrange's equation is: 
\begin{align}\frac{\mathrm{d}}{\mathrm{dt}}\left(\frac{\partial L}{\partial \dot{\bm{n}}_{a}}\right)-\frac{\partial L}{\partial \bm{n}_{a}}=\bm{f}_{n p a},
\label{Lagrange's in na}
\end{align}
where $\bm{f}_{n p a}$ is the non-potential force exerted on free nodal coordinate, and its relation with $\bm{f}_{n p}$ is:
\begin{align}\bm{f}_{n p a}=\bm{E}_{a}^{T} \bm{f}_{n p}.
\label{Lagrange's equ na right}
\end{align}
The relation of left side of Eq. (\ref{Lagrange's in na}) and Eq. (\ref{Lagarangian's equation}) is:
\begin{align}
\frac{\mathrm{d}}{\mathrm{dt}}\left(\frac{\partial L}{\partial \dot{\bm{n}}_{a}}\right) &-\frac{\partial L}{\partial \bm{n}_{a}}=\frac{\partial \bm{n}}{\partial \bm{n}_{a}}\left[\frac{\mathrm{d}}{\mathrm{dt}}\left(\frac{\partial L}{\partial \dot{\bm{n}}}\right)-\frac{\partial L}{\partial \bm{n}}\right] \\
&=\bm{E}_{a}^{T}\left[\frac{\mathrm{d}}{\mathrm{dt}}\left(\frac{\partial L}{\partial \dot{\bm{n}}}\right)-\frac{\partial L}{\partial \bm{n}}\right].
\label{Lagrange's equ na left}
\end{align}
Substitute Eqs. (\ref{Lagrange's equ na right}), (\ref{Lagrange's equ na left}) into Eq. (\ref{Lagrange's in na}), we have the dynamics of tensegrity for the free nodal coordinates:
\begin{align}\bm{E}_{a}^{T}(\bm{M} \ddot{\bm{n}}+\bm{D} \dot{\bm{n}}+\bm{K} \bm{n})=\bm{E}_{a}^{T}\left(\bm{f}_{e x}-\bm{g}\right).
\label{dynamic in na 1}
\end{align}
From the above equation, we can see that the dynamic equation Eq. (\ref{dynamic in na 1}) with constraints is just the $\bm{a}$ rows of Eq. (\ref{dynamics2}). Substitute Eq. (\ref{n=[Ea Eb][na;nb]}) into Eq. (\ref{dynamic in na 1}) and arrange terms related to $\bm{n}_a$ in left side, we obtain Eq. (\ref{dynamics equation reduced order}). We can also have the following form in term of $\ddot{\bm{n}}_{a}$ for programming convenience:
\begin{align}
\ddot{\bm{n}}_{a}=\bm{M}_{a a}^{-1} \bm{E}_{a}^{T}\left(\bm{f}_{e x}-\bm{g}-\bm{M} \bm{E}_{b} \ddot{\bm{n}}_{b}-\bm{D} \dot{\bm{n}}-\bm{K} \bm{n}\right).
\end{align}
\end{proof}

\subsection{Static equilibrium equation}
\label{Section 4}
The static equilibrium equation can be easily obtained from the derived dynamics equation. Here, we give the equilibrium equation in three standard forms and the compatibility equation. The equations developed in this section are useful for the derivation of the linearized dynamics in the next section.  

\begin{theorem}\label{statics}
The three following tensegrity static equilibrium equations are equivalent: \par
1). Tensegrity statics in terms of nodal coordinate vector $\bm{n}$:
\begin{align}
\bm{K} \bm{n} =\bm{f}_{e x}-\bm{g},~
\bm{K} = 
(\bm{C}^T\hat{\bm{x}}\bm{C})\otimes\textbf{I}_3.
\end{align}\par
2). Tensegrity statics in terms of force density vector $\bm{x}$:
\begin{align}
\bm{A}_{1}\bm{x}=\bm{f}_{e x}-\bm{g},~\bm{A}_{1}=\left(\bm{C}^{T} \otimes \textbf{I}_{3}\right) \bm{b.d.}(\bm{H}).
\end{align}\par
3). Tensegrity statics in terms of force vector $\bm{t}$:
\begin{align}\bm{A}_{2} \bm{t}=\bm{f}_{e x}-\bm{g}, ~\bm{A}_{2} =\left(\bm{C}^{T} \otimes \textbf{I}_{3}\right) \bm{b.d.}(\bm{H}) \hat{\bm{l}}^{-1}.
\end{align}
\end{theorem}

\begin{proof}
Let the acceleration part $\ddot{\bm{n}}$ and velocity part $\dot{\bm{n}}$ in Eq. (\ref{dynamics2}) be zeros, the dynamics equation will be reduced into a static equilibrium equation in terms of nodal coordinate vector $\bm{n}$:
\begin{align}\label{static equilibrium1}
\bm{K} \bm{n}=\bm{f}_{e x}-\bm{g}.
\end{align}
This proofs the first statement of Theorem \ref{statics}. \par
Since $\bm{K}$, given in Eq. (\ref{K}), is a function of $\bm{n}$, the product $\bm{Kn}$ is nonlinear in $\bm{n}$. Eq. (\ref{static equilibrium1}) is a nonlinear equilibrium equation. However, the term $\bm{Kn}$ can be also written linearly in terms of force density vector $\bm{x}$:
\begin{align}
\bm{K} \bm{n} &=\left(\bm{C}^{T} \otimes \textbf{I}_{3}\right)\left(\widehat{\bm{x}} \otimes \textbf{I}_{3}\right)\left(\bm{C} \otimes \textbf{I}_{3}\right) \bm{n} \\
&=\left(\bm{C}^{T} \otimes \textbf{I}_{3}\right) \reallywidehat{\left(\textbf{I}_{n_{e}} \otimes \textbf{I}_{3,1} \bm{x}\right)}\left(\bm{C} \otimes \textbf{I}_{3}\right) \bm{n} \\
&=\left(\bm{C}^{T} \otimes \textbf{I}_{3}\right)\reallywidehat{\left(\left(\bm{C} \otimes \textbf{I}_{3}\right) \bm{n}\right)} \textbf{I}_{n_{e}} \otimes \textbf{I}_{3,1} \bm{x} \\
&=\left(\bm{C}^{T} \otimes \textbf{I}_{3}\right) \bm{b.d.}(\bm{H}) \bm{x}.
\label{Kn=~~~x}
\end{align}
Substitute Eq. (\ref{Kn=~~~x}) into Eq. (\ref{static equilibrium1}), we have a linear form of equilibrium equation:
\begin{align}
\bm{A}_{1}\bm{x}=\bm{f}_{e x}-\bm{g},
\label{static equilibrium2}
\end{align}
where $\bm{A}_{1} \in \mathbb{R}^{3 n_n\times n_e}$ is the equilibrium matrix with force density $\bm{x}$ as variable:
\begin{align}\label{A_1}
\bm{A}_{1}=\left(\bm{C}^{T} \otimes \textbf{I}_{3}\right) \bm{b.d.}(\bm{H}),
\end{align}
where $\bm{b.d.}(V)$ is the block diagonal matrix of $V$. This proofs the first and second statements of Theorem \ref{statics} are equivalent. 

The equilibrium equation can also be written linearly in terms of force vector $\bm{t}$ by substitute Eq. (\ref{force density}) into Eq. (\ref{static equilibrium2}):
\begin{align}\bm{A}_{2} \bm{t}=\bm{f}_{e x}-\bm{g},
\label{static equilibrium3}
\end{align}
where $\bm{A}_{2} \in \mathbb{R}^{3 n_n\times n_e}$ is the equilibrium matrix with force vector $\bm{t}$ as variable:
\begin{align}
\bm{A}_{2} =\bm{A}_{1} \hat{\bm{l}}^{-1} =\left(\bm{C}^{T} \otimes \textbf{I}_{3}\right) \bm{b.d.}(\bm{H}) \hat{\bm{l}}^{-1}.
\end{align}
This proofs the second and third statements of Theorem \ref{statics} are equivalent. 
\end{proof}

\subsection{Compatibility equation}

The compatibility equation is the relation between $\mathrm{d}\bm{n}$ and  $\mathrm{d}\bm{l}$ that guarantees the structure deformations are physically valid. The compatibility equation of the \textit{i}th element can be obtained by take the derivative of Eq. (\ref{li}):
\begin{align}l_{i}^{-1} \bm{h}_{i}^{T}\left(\bm{C}_{i} \otimes \textbf{I}_{3}\right) \mathrm{d} \bm{n}= \mathrm{d} l_{i}.
\label{dli}
\end{align}
Stack all the structure element equations in a column, one can obtain: 
\begin{align}\bm{B}_{l}  \mathrm{d} \bm{n}= \mathrm{d} \bm{l},
\label{compatibility equation}
\end{align}
where $\bm{B}_{l}\in \mathbb{R}^{n_e\times 3 n_n }$ is the compatibility matrix of the structure:
\begin{align}\bm{B}_{l}=\hat{\bm{l}}^{-1} \bm{b.d.}(\bm{H})^{T}\left(\bm{C} \otimes \textbf{I}_{3}\right).
\label{Bl}
\end{align}
Note that the compatibility and equilibrium matrix have the following relationship: $\bm{B}_{l}^{T}=\bm{A}_{2}$, which can also be proved by the principle of virtual work.

\section{Linearized tensegrity dynamics}
\label{section 5}
\subsection{Linearized dynamics without constraints}

\begin{theorem}
The finite element linearized tensegrity dynamics with no constraints has the following analytical form: \par
\begin{align}\label{linear_dyn}
\bm{M}\mathrm{d} \ddot{\bm{n}}+\bm{D} \mathrm{d}\dot{\bm{n}}+\bm{K}_T\mathrm{d} \bm{n}=\mathrm{d}\bm{f}_{e x},
\end{align}
where the tangent stiffness matrix $\bm{K}_{T}$ satisfies:
\begin{align}
\bm{K}_{T}=\left(\bm{C}^{T} \widehat{\bm{x}} \bm{C}\right) \otimes \textbf{I}_{3}+\bm{A}_{1} \widehat{\bm{E}_t} \widehat{\bm{A}} \hat{\bm{l}}^{-3} \bm{A}_{1}^{T},
\end{align}
and $\bm{M}$ is given in Eq. (\ref{M}), $D$ is damping matrix, and $f_{ex}$ is external forces on the structure nodes.
\end{theorem}

\begin{proof}

Since the stiffness matrix $\bm{K}$ is a function of nodal coordinate vector $\bm{n}$, the dynamics Eq. (\ref{dynamics2}) is nonlinear. The mass matrix $\bm{M}$ and damping matrix $\bm{D}$ are constant. To linearize the dynamic equation, we can take the total derivative of Eq. (\ref{dynamics2}) and keep the linear terms:
\begin{align}
\bm{M}\mathrm{d} \ddot{\bm{n}}+\bm{D} \mathrm{d}\dot{\bm{n}}+\bm{K}_T\mathrm{d} \bm{n}=\mathrm{d}\bm{f}_{e x}.
\label{linear dynamics}
\end{align}
The tangent stiffness matrix $\bm{K}_T$ can be calculated as:
\begin{align}
\bm{K}_{T}=\left[\frac{\partial(\bm{K} \bm{n})}{\partial \bm{n}}\right]^{T}=\bm{K}+\left[\frac{\partial \bm{x}}{\partial \bm{n}} \frac{\partial(\bm{K} \bm{n})}{\partial \bm{x}}\right]^{T}.
\label{Kt tangent stiffness matrix 0}
\end{align}
The partial derivative of force density vector $\bm{x}$ to nodal coordinate vector $\bm{n}$ can be obtained from Eq. (\ref{force density}):
\begin{align}
\frac{\partial \bm{x}}{\partial \bm{n}} &=\frac{\partial\left[\widehat{\bm{E}_t} \widehat{\bm{A}}\left(\bm{l}_{0}^{-1}-\bm{l}^{-1}\right)\right]}{\partial \bm{n}} \\
&=\frac{\partial \bm{l}}{\partial \bm{n}} \frac{\partial\left(-\bm{l}^{-1}\right)}{\partial \bm{l}} \widehat{\bm{A}} \widehat{\bm{E}_t} \\
&=\bm{B}_{l}^{T} \hat{\bm{l}}^{-2} \widehat{\bm{A}} \widehat{\bm{E}_t} \\
&=\bm{A}_{1} \hat{\bm{l}}^{-3} \widehat{\bm{A}} \widehat{\bm{E}_t}.
\label{x by n}
\end{align}

The derivative of $\bm{Kn}$ with respect to force density $\bm{x}$ is derived from Eq. (\ref{A_1}), then we have: 
\begin{align}\frac{\partial(\bm{K} \bm{n})}{\partial \bm{x}}=\frac{\partial\left(\bm{A}_{1} \bm{x}\right)}{\partial \bm{x}}=\bm{A}_{1}^{T}.
\label{kn by x}
\end{align}
Substitute Eqs. (\ref{x by n}) and (\ref{kn by x}) into Eq. (\ref{Kt tangent stiffness matrix 0}), one can obtain the tangent stiffness matrix $\bm{K}_{T}$:
\begin{align}
\bm{K}_{T}=\left(\bm{C}^{T} \widehat{\bm{x}} \bm{C}\right) \otimes \textbf{I}_{3}+\bm{A}_{1} \widehat{\bm{E}_t} \widehat{\bm{A}} \hat{\bm{l}}^{-3} \bm{A}_{1}^{T}.
\label{Kt tangent stiffness matrix 2}
\end{align}
The first part of Eq. (\ref{Kt tangent stiffness matrix 2}) is usually called the geometry stiffness matrix $\bm{K}_{G}=\left(\bm{C}^{T} \widehat{\bm{x}} \bm{C}\right) \otimes \textbf{I}_{3}$, which is determined by structure topology and force density. The second part is called the material stiffness $\bm{K}_{E}=\bm{A}_{1} \widehat{\bm{E}_t} \widehat{\bm{A}} \hat{\bm{l}}^{-3} \bm{A}_{1}^{T}$, which is governed by structure configuration and elements' axial stiffness.\par
\end{proof}

The linearized dynamics equation can also be written into a standard state space form:
\begin{align}
   \frac{\mathrm{d}}{\mathrm{d}t} \begin{bmatrix}\mathrm{d}\bm{{n}}\\\mathrm{d}\bm{\dot{n}}\end{bmatrix}=\begin{bmatrix}\bm{ 0} & \textbf{I}\\ -\bm{M}^{-1}\bm{K}_T& -\bm{M}^{-1}\bm{D}\end{bmatrix}    \begin{bmatrix}\mathrm{d}\bm{n}\\\mathrm{d}\bm{\dot{n}}\end{bmatrix}+\begin{bmatrix}\bm{0}\\\mathrm{d}\bm{f}_{ex}\end{bmatrix},
\end{align}
which can be used to integrate structure and control designs.

\subsection{Modal analysis of the linearized model with no constraints}

By setting damping matrix $\bm{D}= \bm{0}$ and external force $\bm{f}_{ex}=\bm{0}$ in Eq. (\ref{linear_dyn}), we have the free vibration response of a dynamical system:
\begin{align}
\bm{M}\mathrm{d} \ddot{\bm{n}}+\bm{K}_T\mathrm{d} \bm{n}=\bm{0}.
\label{free vibration}
\end{align}
The solution to the homogeneous Eq.  (\ref{free vibration}) have the following form:
\begin{align} \mathrm{d} \bm{n}=\bm{\varphi} \sin (\omega t-\theta),
\label{dn}\end{align}
which represents a periodic response with a typical frequency $\omega$. Substitute Eq. (\ref{dn}) into Eq. (\ref{free vibration}), we have:
\begin{align}
\left(\bm{K}_{T}-\omega^{2} \bm{M}\right) \bm{\varphi} \sin (\omega t-\theta)=\bm{0},
\end{align}
and since $\sin (\omega t-\theta)\ne0$ for most times, we have:
\begin{align}
\bm{K}_{T} \bm{\varphi}=\omega^{2}\bm{M} \bm{\varphi},
\label{eigenvalue problem}
\end{align}
which is a standard eigenvalue problem. The $\omega$ is known as the natural frequency of the system and $\bm{\varphi}$ is the corresponding mode.

\subsection{Linearized dynamics with constraints}

\begin{theorem}
The finite element linearized tensegrity dynamics with constraints has the following analytical form: 
\begin{align}\nonumber
& \bm{M}_{a a} \mathrm{d}\ddot{\bm{n}}_{a} + \bm{D}_{a a} \mathrm{d}\dot{\bm{n}}_{a}+\bm{K}_{Ta a} \mathrm{d}\bm{n}_{a}\\ &= \bm{E}_{a}^{T} \mathrm{d}\bm{f}_{e x} -\bm{M}_{a b} \mathrm{d}\ddot{\bm{n}}_{b}-\bm{D}_{a b} \mathrm{d}\dot{\bm{n}}_{b} -\bm{K}_{Tab} \mathrm{d}\bm{n}_{b},
\end{align}
where the tangent stiffness matrix $\bm{K}_{T_{aa}}$,$\bm{K}_{T_{ab}}$ is:
\begin{align}
 \bm{K}_{Taa} & =\bm{E}_{a}^{T} \bm{K}_T \bm{E}_{a},~\bm{K}_{Tab} =\bm{E}_{a}^{T} \bm{K}_T \bm{E}_{b},
\end{align}
$\bm{M}_{a a}$, $\bm{M}_{a b}$, $\bm{D}_{a a}$, $\bm{D}_{a b}$ are given in Eqs. (\ref{aa_ab1}) - (\ref{aa_ab3}), and $\bm{K}_T$ is given in Eq. (\ref{Kt tangent stiffness matrix 2}).
\end{theorem}

\begin{proof}

To linearize the dynamic equation considering constraints, we can take the total derivative of Eq. (\ref{dynamic in na 1}) and keep the linear terms:
\begin{align}
\bm{E}_{a}^{T}(\bm{M}\mathrm{d} \ddot{\bm{n}}+\bm{D} \mathrm{d}\dot{\bm{n}}+\bm{K}_T\mathrm{d} \bm{n})=\bm{E}_{a}^{T}\mathrm{d}\bm{f}_{e x}.
\label{linear dynamics with constraints}
\end{align}
Substitute Eq.(\ref{n=[Ea Eb][na;nb]}) into Eq.  (\ref{linear dynamics with constraints}), we have:
\begin{align}\nonumber
& \bm{M}_{a a} \mathrm{d}\ddot{\bm{n}}_{a} + \bm{D}_{a a} \mathrm{d}\dot{\bm{n}}_{a}+\bm{K}_{Ta a} \mathrm{d}\bm{n}_{a}\\ &= \bm{E}_{a}^{T} \mathrm{d}\bm{f}_{e x} -\bm{M}_{a b} \mathrm{d}\ddot{\bm{n}}_{b}-\bm{D}_{a b} \mathrm{d}\dot{\bm{n}}_{b} -\bm{K}_{Tab} \mathrm{d}\bm{n}_{b},
\label{linear dynamics with constraints2}
\end{align}
in which $\bm{K}_{Ta a}$ and $\bm{K}_{Ta b}$ are given by:
\begin{align}
    \bm{K}_{Taa} & =\bm{E}_{a}^{T} \bm{K}_T \bm{E}_{a},~\bm{K}_{Tab}=\bm{E}_{a}^{T} \bm{K}_T \bm{E}_{b}.
\end{align}
\end{proof}

Similarly, one can write the linearized dynamics equation with constraints into a state space form:
\begin{equation}
\begin{aligned}
\frac{\mathrm{d}}{\mathrm{d}t} \begin{bmatrix}\mathrm{d}\bm{n}_a\\\mathrm{d}\dot{\bm{n}_{a}}\end{bmatrix}=\begin{bmatrix}\bm{ 0}  & \textbf{I}\\ -\bm{M}_{aa}^{-1}K_{Ta a} & \bm{-M}_{aa}^{-1} \bm{D}_{aa}\end{bmatrix}    \begin{bmatrix}\mathrm{d}\bm{n}_a\\\mathrm{d}\dot{\bm{n}}_a\end{bmatrix}\\+\begin{bmatrix}\bm{0}\\\bm{E}_{a}^{T}\mathrm{d}\bm{f}_{ex}-\bm{M}_{a b} \mathrm{d}\ddot{\bm{n}}_{b}-\bm{D}_{a b} \mathrm{d}\dot{\bm{n}}_{b} -\bm{K}_{Tab} \mathrm{d}\bm{n}_{b}\end{bmatrix},
\end{aligned}
\label{dynamic equation with constraints in first order}
\end{equation}
as an interface to integrate structure and control designs.



\subsection{Modal analysis of the linearized model with constraints}
Similarly, for tensegrity dynamics with constraints, the free vibration response can be obtained from Eq. (\ref{linear dynamics with constraints2}) by neglecting damping and external force:
\begin{align}
\bm{M}_{aa}\mathrm{d} \ddot{\bm{n}}_a+\bm{K}_{Taa}\mathrm{d} \bm{n}_a=\bm{0}.
\label{free vibration with constraints}
\end{align}
By similar derivation from Eq. (\ref{free vibration}) to Eq. (\ref{eigenvalue problem}), the eigenvalue problem of tensegrity with constraints is given by:
\begin{align}
\bm{K}_{Taa} \bm{\varphi}=\omega^{2}\bm{M}_{aa} \bm{\varphi},
\label{eigenvalue problem with constratints}
\end{align}
where $\omega$ is the natural frequency of the system and $\bm{\varphi}$ is the corresponding mode.

\section{Numerical examples}
\label{section 6}
We believe a general dynamics should be capable of conducting these kinds of studies for any tensegrity structures: 1. Rigid body dynamics with acceptable errors (by setting relatively high stiffness for bars in the FEM simulation). 2. Finite element method (FEM) dynamics that allow bars and strings to have elastic or plastic deformations. 3. The dynamics should allow various kinds of boundary conditions, for example, nodes are fixed or in the presence of static or dynamic external forces (i.e., gravitational force, some specified forces, or arbitrary seismic vibrations, etc.). 4. Accurate modal analysis, including natural frequency and corresponding modes. 

Thus, three dynamic examples (a double pendulum, a cantilever truss with external force, and seismic analysis of a tensegrity tower) are carefully selected and studied to verify the proposed nonlinear tensegrity FEM dynamics (we call it TsgFEM). The obtained results are compared with other dynamics simulation methods, including analytical results and commercial FEM software ANSYS.


\subsection{Example 1: Dynamics of a double pendulum}



This example is chosen to check if the structure behaves close to rigid body dynamics with an acceptable error if we use high stiffness materials for bars. The time history of nodal positions will be compared with analytical results from rigid body dynamics. 
The initial configuration of a double pendulum is shown in Fig.\ref{db_config}, and node 1 is fixed to the wall. The two bars have same length $l=1$ m, same mass $m=1$ kg, and same hanging angle $\theta_1=\theta_2=45^\circ$. The cross-sectional area and Young's modulus are $A=10^{-4}$ m$^2$ and $E=2.06\times10^{11}$ Pa, respectively. In the analysis of the dynamics, the time step and total simulation time are chosen to be $\Delta t=5\times10^{-5}$ s and $t=5$ s. There is no damping, and only gravitational force is considered as the external force.

Fig.\ref{Nodal position of the double pendulum} shows the time history of nodal position by TsgFEM. Fig.\ref{Error of Nodal position compared with RBD} gives the error of the nodal position between TsgFEM and the analytical solution obtained from rigid body dynamics, whose analytical equations are derived in the Appendix \ref{app_pendulum}. Fig.\ref{Error of bar length compared with RBD} is the error of bar length between TsgFEM and rigid body dynamics. From these figures, we can see that the error of bar length ($10^{-7}\sim10^{-6}$ m) and error of nodal coordinates ($10^{-5}\sim10^{-4}$ m) are relatively small and oscillate periodically with high frequency. This is reasonable because the strain is allowed in TsgFEM, and the high-frequency oscillation is caused by the high axial stiffness of the bar. The TsgFEM can capture the periodic elongation of bars if the time step is properly chosen. In signal processing of Nyquist rate, the sampling frequency should be at least two times of signal frequency, and normally engineers use 5 \textasciitilde~10 times, we choose a time step of $\Delta t={\mathrm{T}_{min}}/{8}={\pi}/{4\omega_{max}}$, where ${\mathrm{T}_{min}}$ is the shortest period corresponding to the highest natural frequency calculated by Eq. (\ref{eigenvalue problem with constratints}) to capture the highest vibration mode of the bar as well as guarantee the convergence in solving the dynamics equation.



\begin{figure}
    \centering
    \includegraphics[scale=0.5]{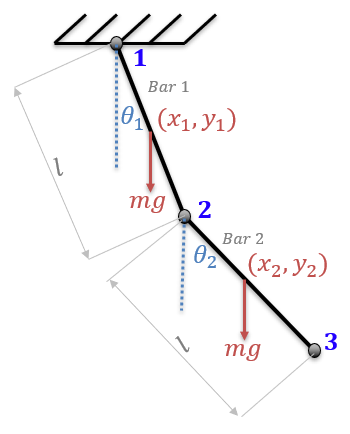}
    \caption{Schematic diagram of a double pendulum in the presence of gravity. The two bars have same mass and length.}
    \label{db_config}
\end{figure}


\begin{figure}
    \centering
    \includegraphics[scale=0.55]{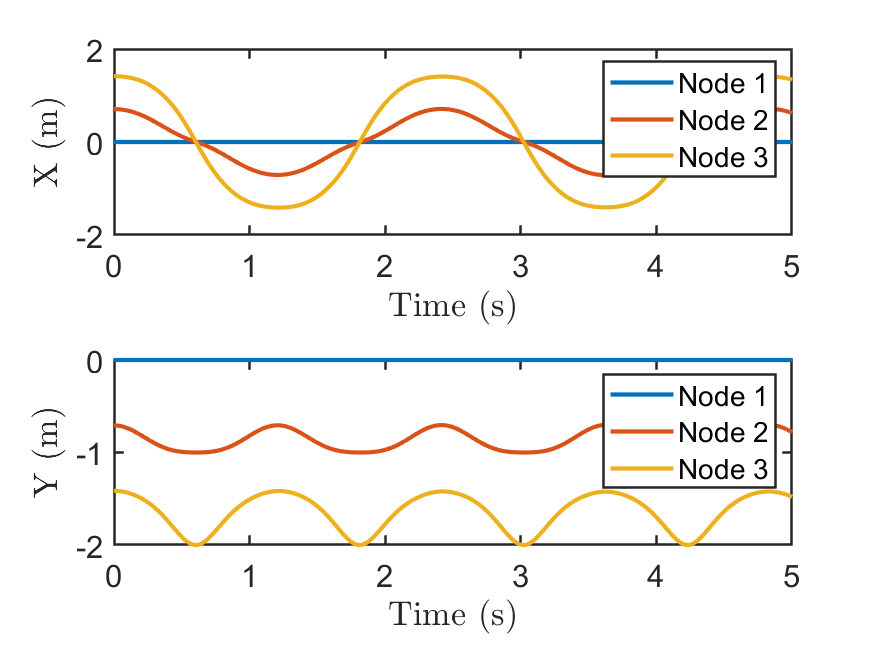}
    \caption{$X$- and $Y$-coordinate time histories of node 1, 2, and 3 of the double pendulum.}
    \label{Nodal position of the double pendulum}
\end{figure}

\begin{figure}
    \centering
    \includegraphics[scale=0.55]{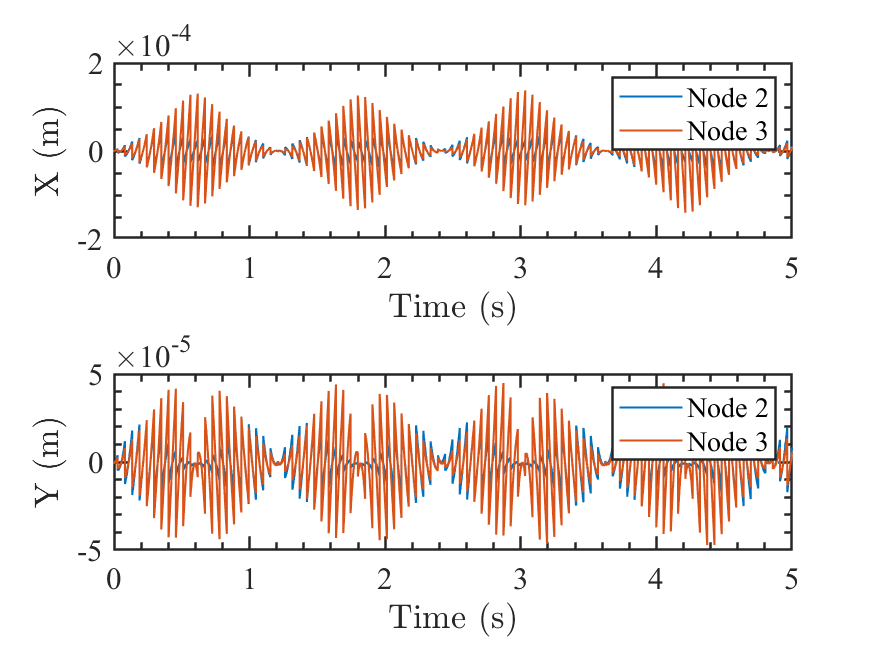}
    \caption{$X$- and $Y$-coordinate error time histories of node 2 and 3 between the TsgFEM dynamics and rigid body dynamics .}
    \label{Error of Nodal position compared with RBD}
\end{figure}

\begin{figure}
    \centering
    \includegraphics[scale=0.55]{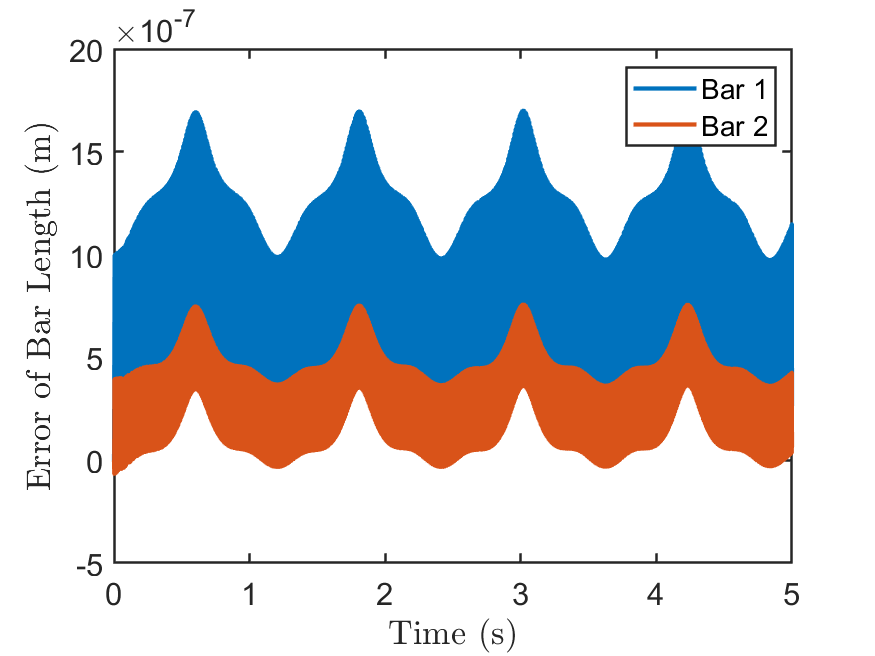}
    \caption{Bar length errors of the two bars in the double pendulum by TsgFEM simulation.}
    \label{Error of bar length compared with RBD}
\end{figure}

\subsection{Example 2: Cantilever truss in external force}

\begin{figure}
    \centering
    \includegraphics[scale=0.30]{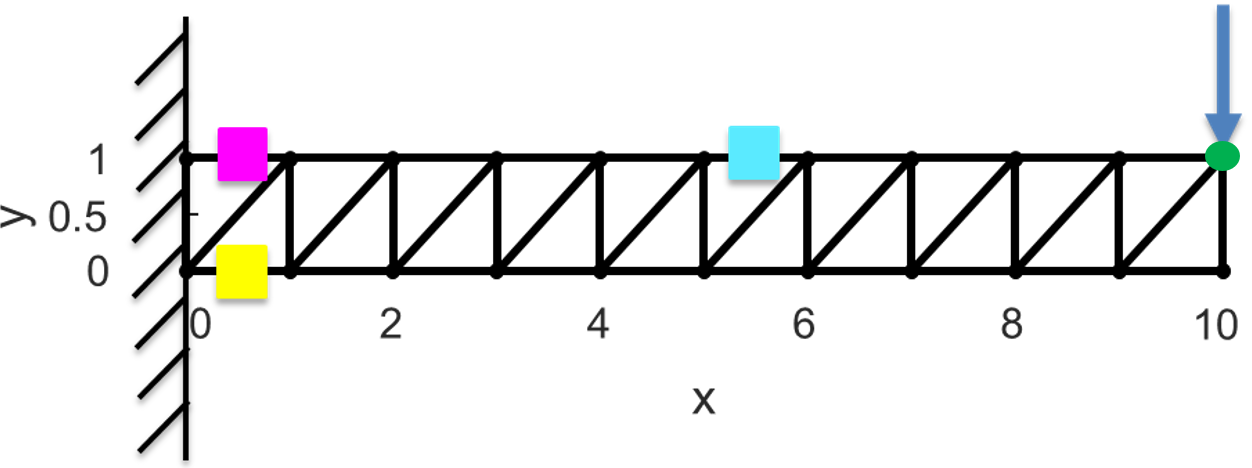}
    \caption{Configuration of a planer truss in the presence of a step load P with the left two nodes fixed to a wall in the given direction. The aspect ratio (length over width) of the truss is 10:1. We examine the strain-stress time histories of the purple, yellow, and blue elements marked by square blocks as well the Y-coordinate time history of the green dot.}
    \label{Configuration of a planer truss}
\end{figure}


This example is selected to verify the proposed dynamics method is capable of doing modal analysis as well as conducting time history analysis of structures with linear-elastic, multilinear elastic, and plastic materials. The natural frequency, mode shapes, and time history information of the structure obtained by TsgFEM will be compared with commercial FEM software ANSYS.

Fig.\ref{Configuration of a planer truss} shows a 10 m $\times$ 1 m planer truss, the left two nodes at the wall are fixed. A step load $P=1\times10^5$ N in the downward direction of the Y-axis is exerted on the green dot at time $t=0$ s. Cross-section area of each element is $A$ = 0.0025 m$^2$. Young's modulus of linear elastic material is $E=2.0604\times10^{11}$ Pa. The  multilinear elastic material is defined by two points ($1.456\times10^{-6}$, 300 MPa), ($2.333\times10^{-6}$, 435 MPa) in the piece-wise stress-strain curve. The elastoplastic material uses a bi-linear kinematic hardening plasticity model, in which Young's modulus is $E=2.0604\times10^{11}$ Pa, the yield stress is 300 MPa, and the tangent modulus in plastic is $E_t=6.1799\times10^{9}$ Pa. Damping and gravitational force is not considered in this example. Time step is $\Delta t=10^{-4}$ s, and total analysis time is $t=1$ s.

\begin{figure}
    \centering
    \includegraphics[scale=0.6]{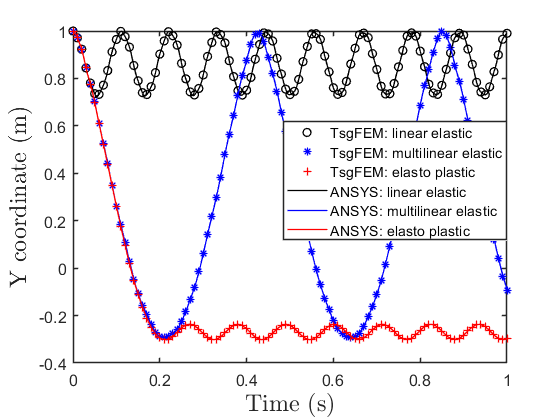}
    \caption{The Y-coordinate time history of the green dot in Fig.\ref{Configuration of a planer truss} by using three kinds of materials: linear elastic, multilinear elastic, and elasto plastic. And a comparison of nodal coordinate time histories between TsgFEM and ANSYS.}
    \label{Comparison of FEM with ANSYS truss}
\end{figure}

\begin{figure}
    \centering
    \includegraphics[scale=0.25]{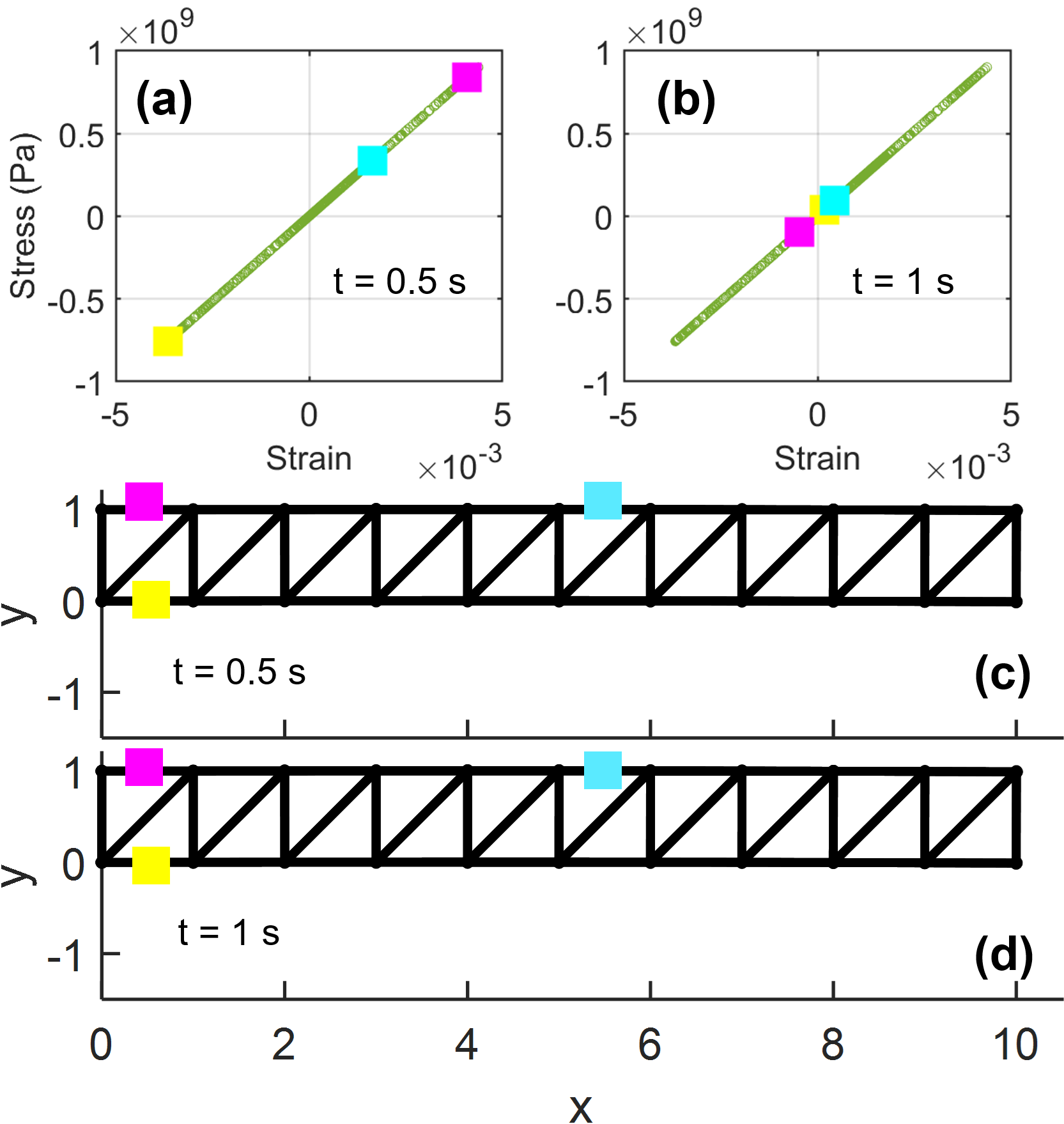}
    \caption{For linear elastic material of structure members, (a) and (b) are stress-strain of the purple, yellow, and blue blocks at t =  0.5 s and t = 1 s. (c) and (d) are corresponding structure deformation at t =  0.5 s and t = 1 s.}
    \label{ss_linear_elastic}
\end{figure}

\begin{figure}
    \centering
    \includegraphics[scale=0.25]{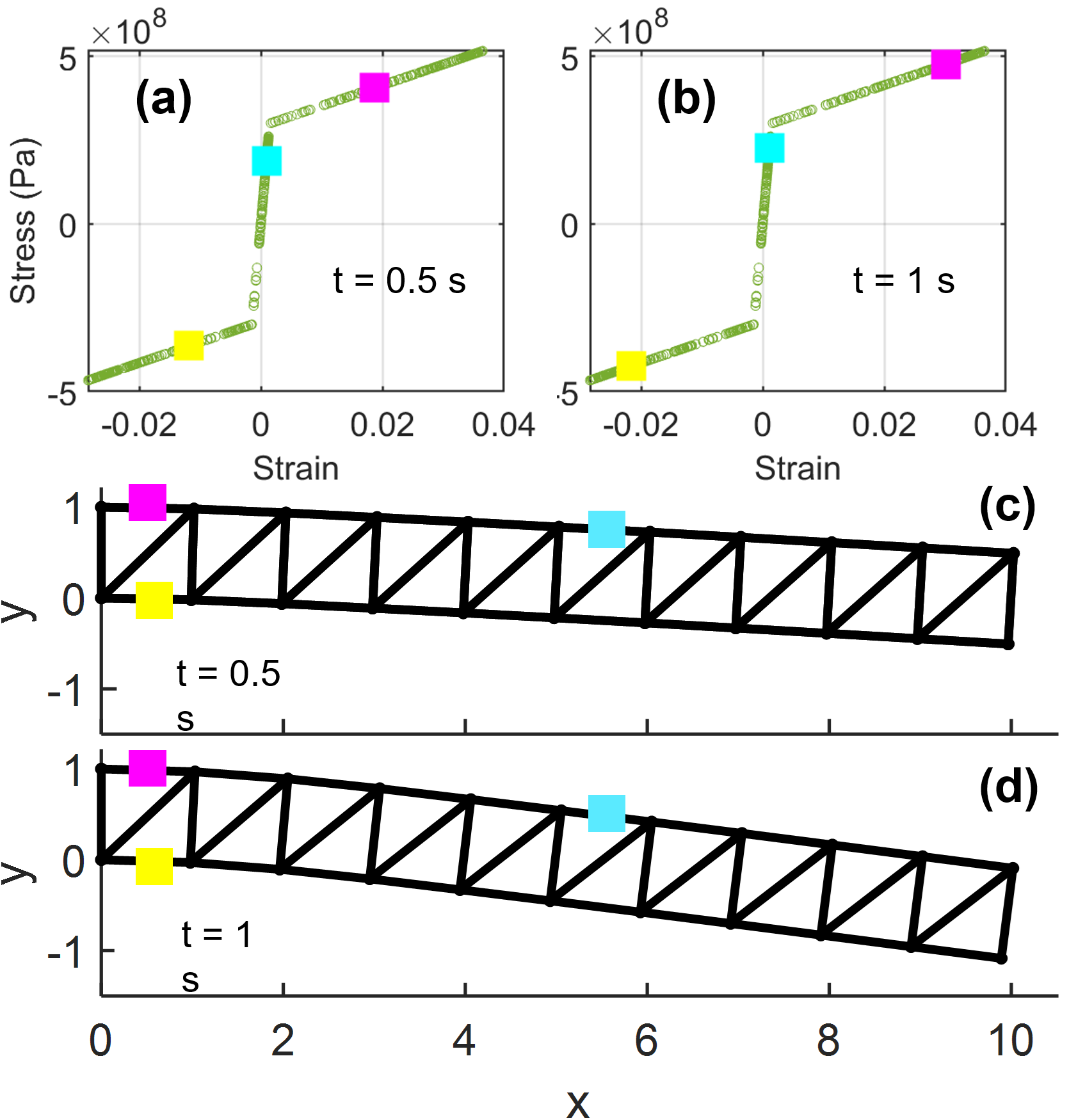}
    \caption{For multilinear elastic material of structure members, (a) and (b) are stress-strain of the purple, yellow, and blue blocks at t =  0.5 s and t = 1 s. (c) and (d) are corresponding structure deformation at t =  0.5 s and t = 1 s.}
    \label{ss_multilinear_elastic}
\end{figure}

\begin{figure}
    \centering
    \includegraphics[scale=0.25]{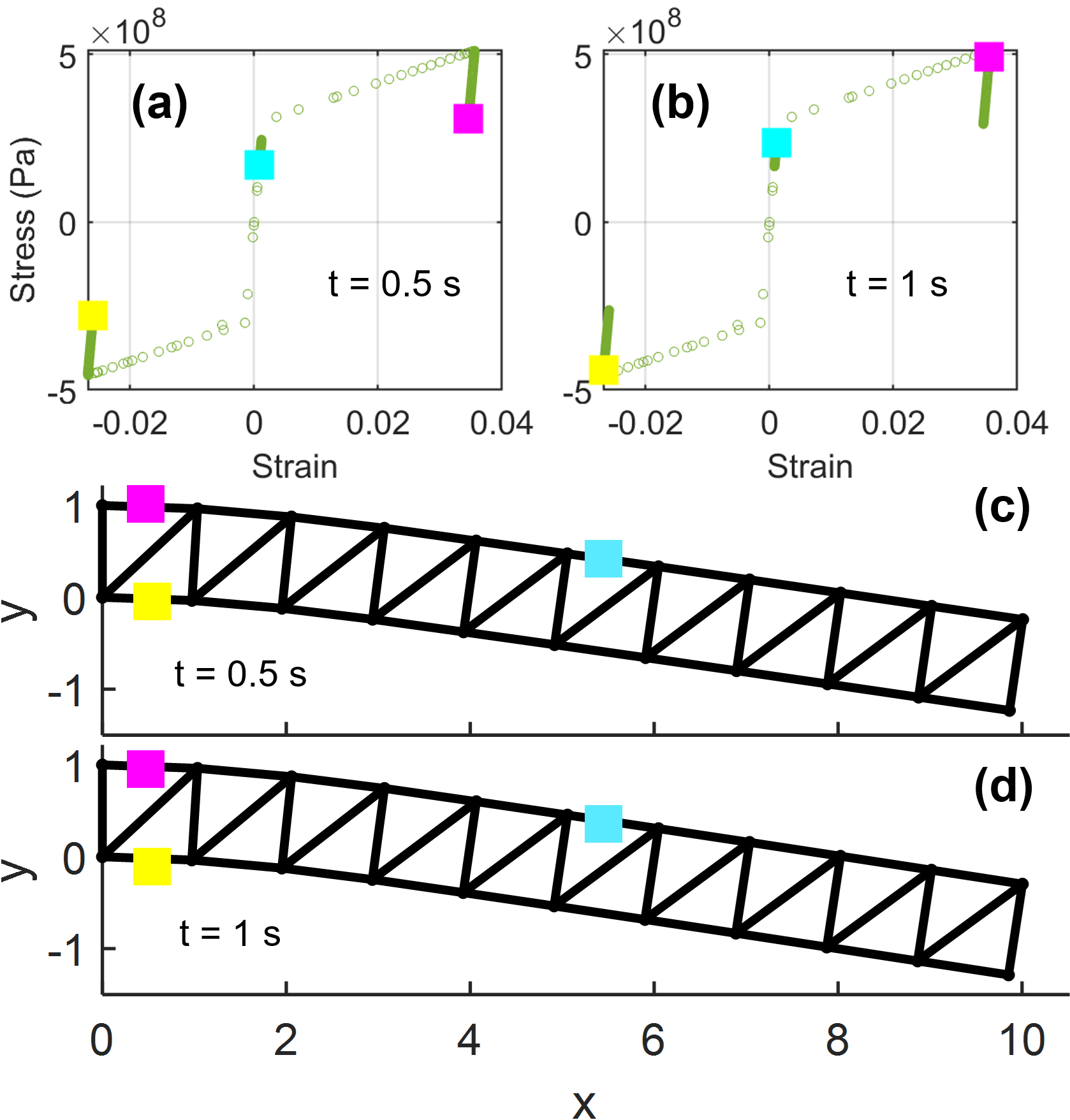}
    \caption{For plastic material of structure members, (a) and (b) are stress-strain of the purple, yellow, and blue blocks at t =  0.5 s and t = 1 s. (c) and (d) are corresponding structure deformation at t =  0.5 s and t = 1 s.}
    \label{ss_plastic}
\end{figure}

The dynamic response of the Y-coordinate of the green node for different materials and comparison with ANSYS in Fig.\ref{Configuration of a planer truss} are compared and given in Fig.\ref{Comparison of FEM with ANSYS truss}. In ANSYS, the transient analysis is used with a consistent mass matrix. The average errors of y-coordinate of node H between TsgFeM and ANSYS with linear elastic, multilinear elastic, and elastoplastic material are $0.09\%$, $0.53\%$, and $0.06\%$, respectively. For the three kinds of materials, the stress-strain of the purple, yellow,and blue blocks and corresponding structure deformation in Fig.\ref{Configuration of a planer truss} at at t = 0.5 s and t = 1 s are given in Figs.\ref{ss_linear_elastic}-\ref{ss_plastic}.\par

\begin{figure}
    \centering
    \includegraphics[scale=0.58]{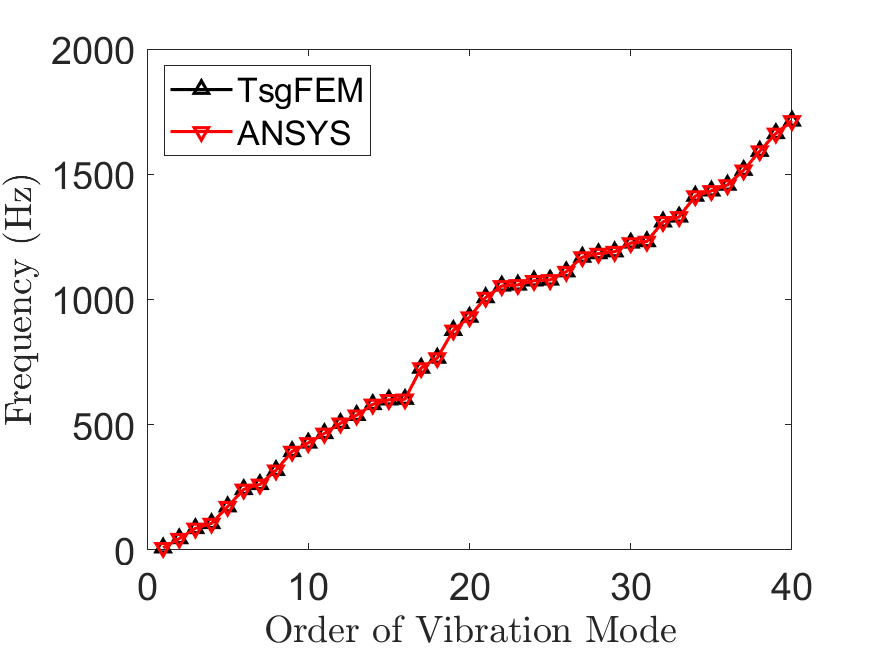}
    \caption{Natural frequencies of the planer truss with respect to the order of vibration mode by TsgFEM and ANSYS. Since the left two nodes of the planar truss is fixed, there are 20 free nodes (40 DOF) in the structure. The number of order of vibration modes is 40.}
    \label{frequency of truss}
\end{figure}

\begin{figure}
    \centering
    \includegraphics[scale=0.5]{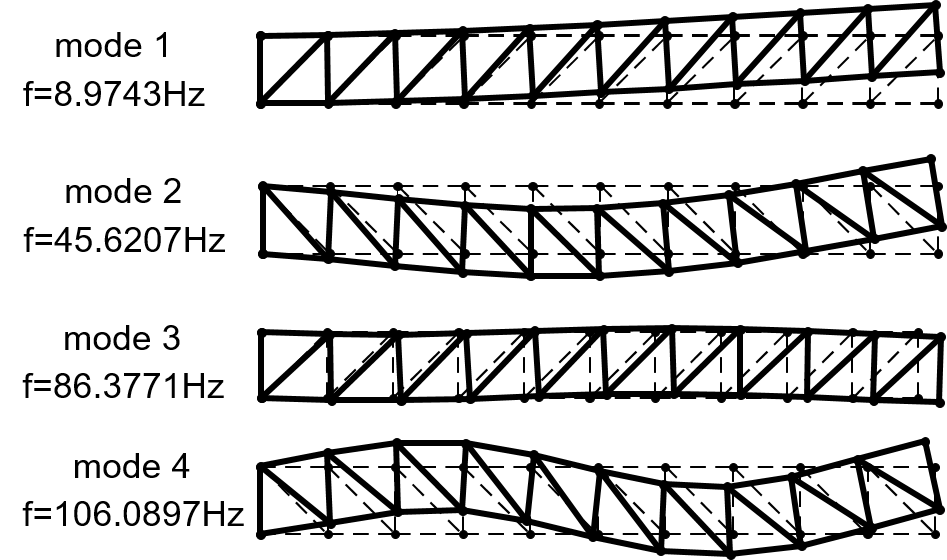}
    \caption{The first four mode shapes of the truss obtained by TsgFEM. The solid line sub-figures from top to bottom are: mode 1, f = 8.9734 Hz; mode 2, f = 45.6159 Hz; mode 3, f = 86.3678 Hz; and mode 4, f = 106.0784 Hz. The dotted lines under the solid lines are shapes of the original truss structure.}
    \label{mode shape of the planer truss TsgFEM}
\end{figure} 

\begin{figure}
    \includegraphics[scale=0.335]{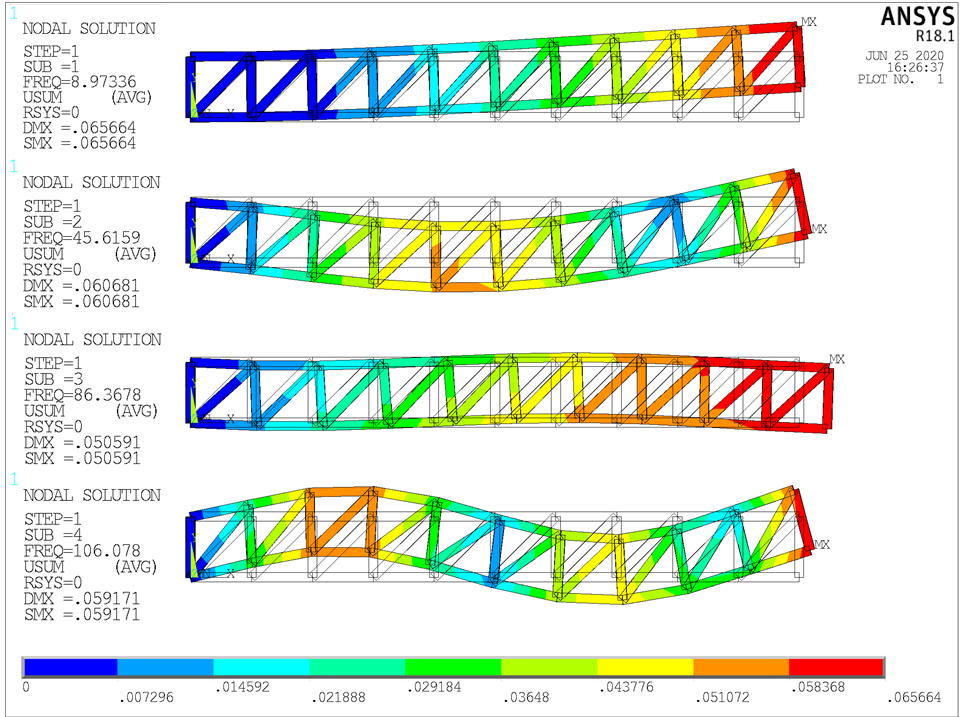}
    \caption{The first four mode shapes of the truss obtained by ANSYS. The colorful sub-figures from left to right are: mode 1, f = 8.9734 Hz; mode 2, f = 45.6159 Hz; mode 3, f = 86.3678 Hz; and mode 4, f = 106.0784 Hz. The light grey lines attached with each colorful mode shapes are the original truss structures. One can also see the natural frequencies obtained from ANSYS on the left side of this plot. }
    \label{Mode shape of truss by ansys}
\end{figure}

The natural frequency and first four modes of the cantilever truss calculated by TsgFEM are compared with ANSYS, as shown in  Figs.\ref{frequency of truss}-\ref{Mode shape of truss by ansys}. The comparative error of frequency by the two methods is $3.6640\times 10^{-13}$.



\subsection{Example 3: Seismic simulation of a tensegrity tower}


A double layer prism is picked as an example to verify the proposed dynamics approach has the ability to do modal analysis and seismic simulation of tensegrity structures with pinned nodes constraints. The double-layer prism is 30 m high, and the radius of the circumscribed circle is 10 m. The structure has two integral prestress modes by making a group of members in rotational symmetric positions \cite{yuan2003integral,zhang2017initial}. The prestress is determined by assigning the prestress of two groups of members. For example, the compression force of bars in two layers is $1.0 \times10^5$ N \cite{ma2019shape}. The cross-sectional area is designed by $10\%$ of yielding and buckling stress in minimal mass design \cite{skelton2009tensegrity}. 

Fig.\ref{Mode1_tower} and \ref{Mode1_tower_ansys} are the first four mode shapes obtained by TsgFEM and ANSYS. Fig.\ref{frequency_tower} is the comparison of the frequency of all 18 modes, and the comparative error of frequency by the two methods is about $4.9457\times 10^{-5}$.\par
The seismic simulation is solved by Eq. (\ref{dynamic equation with constraints in first order}) with sine seismic wave $x = 5sin(4\pi t)$ m exerted on ground motion in the X-direction. The time history of member force of a horizontal string in the first stage calculated by TsgFEM and commercial software ANSYS is shown in Fig.\ref{Time history a member's force}. The average error between TsgFEM and ANSYS in linear elastic members and linear elastic members considering slack of string is $2.91\%$, $2.75\%$, respectively.


\begin{figure}
    \centering
    \includegraphics[scale=0.5]{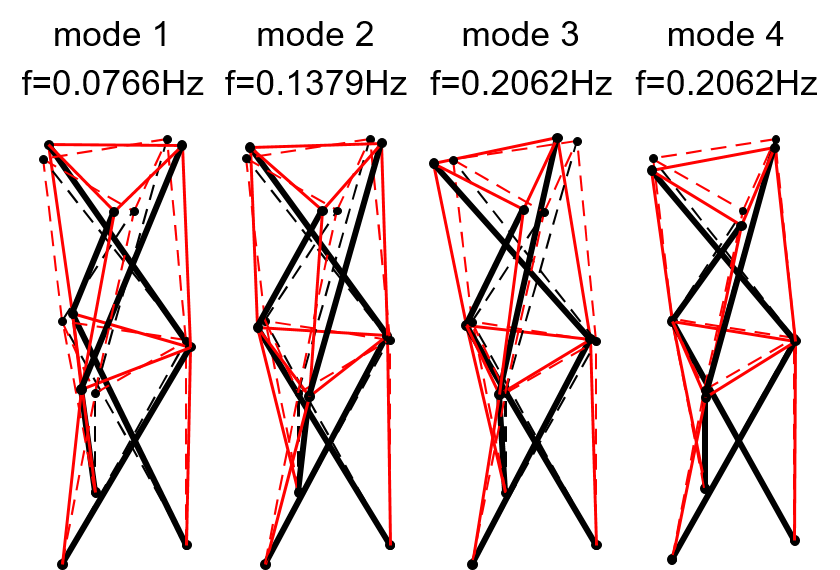}
    \caption{The first four mode shapes of the prism tower obtained by TsgFEM. The solid line sub-figures from top to bottom are: mode 1, f = 0.0766 Hz; mode 2, f = 0.1379 Hz; mode 3, f = 0.2062 Hz; and mode 4, f = 0.2062 Hz. The dotted lines under the solid lines are shapes of the original prism tower structure.}
    \label{Mode1_tower}
\end{figure}

\begin{figure}
    \centering
    \includegraphics[scale=0.35]{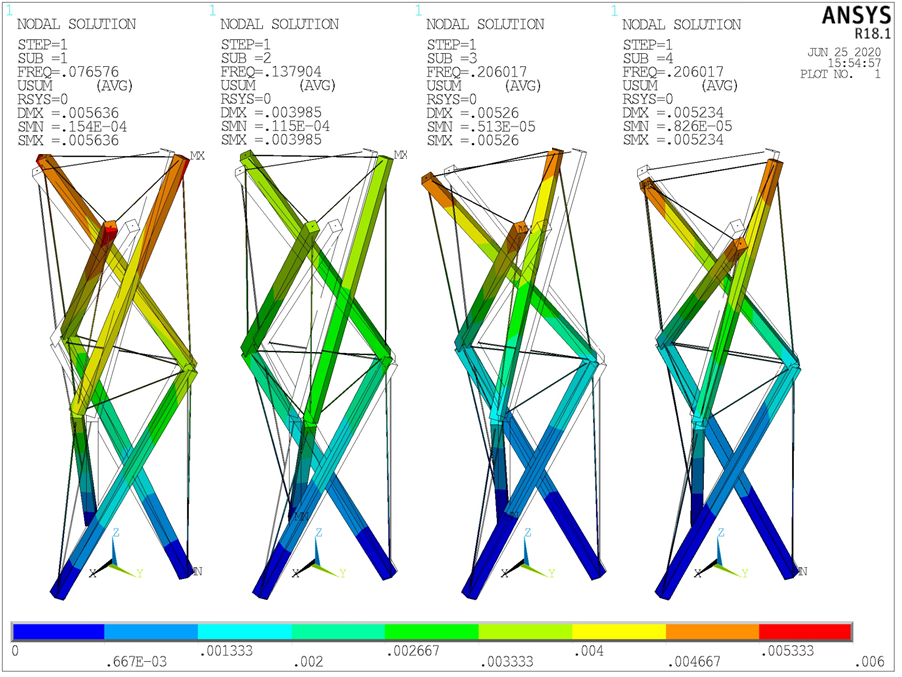}
    \caption{The first four mode shapes of the prism tower obtained by ANSYS. The colorful sub-figures from left to right are: mode 1, f = 0.0766 Hz; mode 2, f = 0.1379 Hz; mode 3, f = 0.2062 Hz; and mode 4, f = 0.2062 Hz. The grey lines under the colorful lines are shapes of the original prism tower structure.}
    \label{Mode1_tower_ansys}
\end{figure}

\begin{figure}
    \centering
    \includegraphics[scale=0.60]{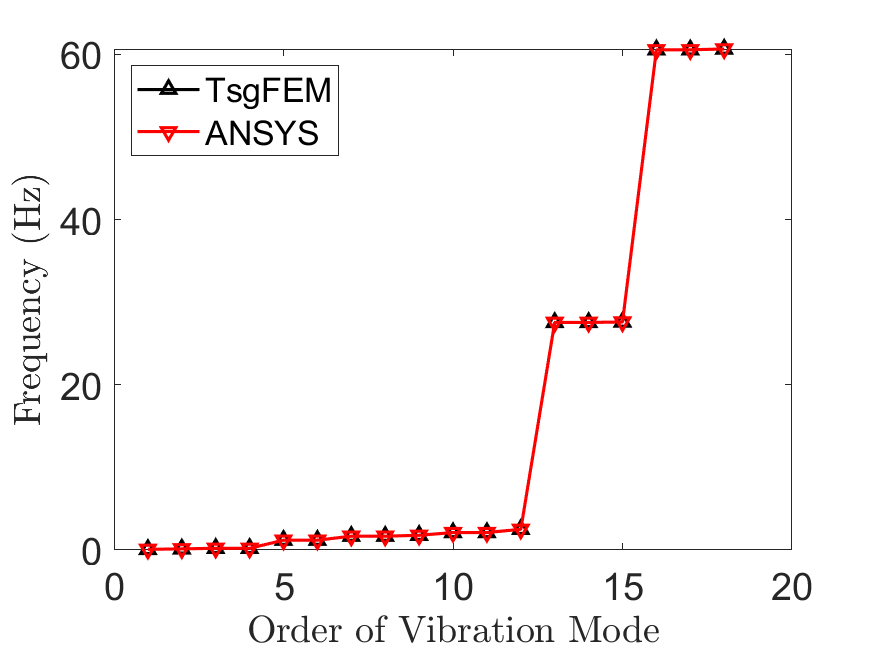}
    \caption{Natural frequencies of the prism tower with respect to the order of vibration mode by TsgFEM and ANSYS. Since the bottom three nodes of the 3D prism tower is fixed, there are 6 free nodes (18 DOF) in the structure. The number of order of vibration modes is 18.}
    \label{frequency_tower}
\end{figure}

\begin{figure}
    \centering
    \includegraphics[scale=0.60]{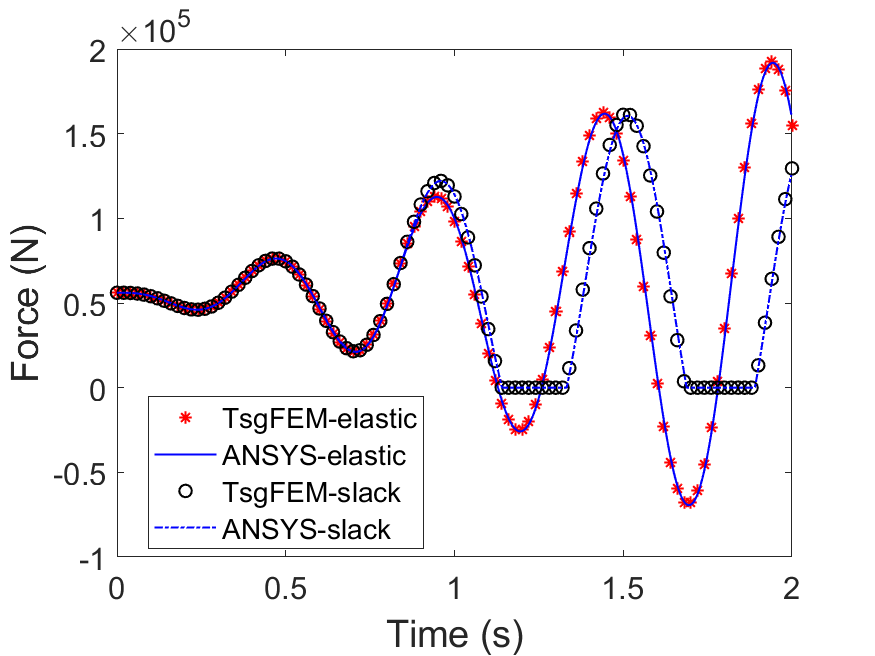}
    \caption{Time histories comparison of forces in the vertical string between TsgFEM and ANSYS with two kinds of structure materials: elastic and slack.}
    \label{Time history a member's force}
\end{figure}

\section{Conclusion}
A finite element analysis approach to non-linear tensegrity dynamics based on the Lagrangian method with a nodal coordinate vector as the variable is given in this paper. This approach allows one to conduct comprehensive studies on any tensegrity systems with any node constraints and various load conditions (i.e., gravitational force, some specified forces, and arbitrary seismic vibrations). Results show that this method is very accurate compared with analytical solutions of rigid body dynamics and FEM software ANSYS. For example, in the double pendulum simulation, the bar length error is $10^{-7}\sim10^{-6}$ m. In the truss example, comparative frequency errors are $3.6640 \times 10^{-13}$, the average comparative error of node coordinates in the linear elastic, multilinear elastic, and elastoplastic material simulation is 0.12\%, 0.50\%, and 0.02\% compared with ANSYS. In the seismic simulation, the comparative frequency error is $4.9457 \times 10^{-5}$, the average comparative error of node coordinates in the linear elastic and multilinear elastic material simulation is 2.91\%, 2.75\% compared with ANSYS. The accurate linearized model in the state space form can be an interface to integrate other disciplines, such as control and signal processing. This study paves a way to perform accurate tensegrity simulations as well as comprehensive understandings of the performance of both structures and materials. 


\label{section 7}


\section{Appendix}
\subsection{Dynamics of the double pendulum}\label{app_pendulum}

Let the mass of the two bars in the double pendulum be $m_1 = m_2 = m$ with a length of $l_1=l_2=l$, from geometric properties shown in Fig.\ref{db_config}, one can write,
\begin{align}
\mathrm{x}_{1}& = \frac{l}{2} \sin {\theta_{1}},~ y_{1}=-\frac{l}{2} \cos{\theta_{1}}, \\
x_{2} & =l(\sin{\theta_{1}}+\frac{1}{2} \sin{\theta_{2}}), \\
y_{2} & =-l(\cos {\theta_{1}}+\frac{1}{2} \cos{\theta_{2}}).
\end{align}
Define $\mathrm{L} = \mathrm{T} - \mathrm{V}$, where $\mathrm{T}$ and $\mathrm{V}$ are kinetic energy and potential energy of the system, then:
\begin{align} \nonumber
\mathrm{L}= \frac{m}{2}\left(\dot{x}_{1}^{2}+\dot{y}_{1}^{2}+\dot{x}_{2}^{2}+\dot{y}_{2}^{2}\right) & + \frac{1}{2} I\left(\dot{\theta}_{1}^{2}+\dot{\theta}_{2}^{2}\right) \\ & -m g\left(y_{1}+y_{2}\right),
\end{align}
where $I = \frac{1}{12}ml^2$ is moment of inertia about the center of mass of the bar. Using Lagrangian method, and $m = 1$ kg and $l= 1$ m, we get:
\begin{align}\nonumber
 8 \ddot{\theta}_{1}+3 \ddot{\theta}_{2} \cos \left(\theta_{1}-\theta_{2}\right)  + & 3 \dot{\theta}_{2}^{2} \sin \left(\theta_{1}-\theta_{2}\right) \\ & +9 \frac{g}{l} \sin \left(\theta_{1}\right)=0, \\ \nonumber
 2 \ddot{\theta}_{2}+3 \ddot{\theta}_{1} \cos \left(\theta_{1}-\theta_{2}\right) - & 3 \dot{\theta}_{1}^{2} \sin \left(\theta_{1}-\theta_{2}\right) \\& +3 \frac{g}{l} \sin \left(\theta_{2}\right)=0.
\end{align}

\printcredits

\bibliographystyle{cas-model2-names}

\bibliography{cas-refs}

\begin{thebibliography}{40}
\expandafter\ifx\csname natexlab\endcsname\relax\def\natexlab#1{#1}\fi
\providecommand{\url}[1]{\texttt{#1}}
\providecommand{\href}[2]{#2}
\providecommand{\path}[1]{#1}
\providecommand{\DOIprefix}{doi:}
\providecommand{\ArXivprefix}{arXiv:}
\providecommand{\URLprefix}{URL: }
\providecommand{\Pubmedprefix}{pmid:}
\providecommand{\doi}[1]{\href{http://dx.doi.org/#1}{\path{#1}}}
\providecommand{\Pubmed}[1]{\href{pmid:#1}{\path{#1}}}
\providecommand{\bibinfo}[2]{#2}
\ifx\xfnm\relax \def\xfnm[#1]{\unskip,\space#1}\fi
\bibitem[{Bathe(2007)}]{bathe2007finite}
\bibinfo{author}{Bathe, K.J.}, \bibinfo{year}{2007}.
\newblock \bibinfo{title}{Finite element method}.
\newblock \bibinfo{journal}{Wiley encyclopedia of computer science and
  engineering} , \bibinfo{pages}{1--12}.
\bibitem[{Cefalo and Mirats-Tur(2011)}]{cefalo2011comprehensive}
\bibinfo{author}{Cefalo, M.}, \bibinfo{author}{Mirats-Tur, J.M.},
  \bibinfo{year}{2011}.
\newblock \bibinfo{title}{A comprehensive dynamic model for class-1 tensegrity
  systems based on quaternions}.
\newblock \bibinfo{journal}{International journal of solids and structures}
  \bibinfo{volume}{48}, \bibinfo{pages}{785--802}.
\bibitem[{Chen et~al.(2020a)Chen, Goyal, Majji and Skelton}]{chen2020habitat}
\bibinfo{author}{Chen, M.}, \bibinfo{author}{Goyal, R.},
  \bibinfo{author}{Majji, M.}, \bibinfo{author}{Skelton, R.E.},
  \bibinfo{year}{2020}a.
\newblock \bibinfo{title}{Design and analysis of a growable artificial gravity
  space habitat}.
\newblock \bibinfo{journal}{Aerospace Science and Technology} ,
  \bibinfo{pages}{106147}.
\bibitem[{Chen et~al.(2020b)Chen, Liu and Skelton}]{chen2020design}
\bibinfo{author}{Chen, M.}, \bibinfo{author}{Liu, J.},
  \bibinfo{author}{Skelton, R.E.}, \bibinfo{year}{2020}b.
\newblock \bibinfo{title}{Design and control of tensegrity morphing airfoils}.
\newblock \bibinfo{journal}{Mechanics Research Communications} ,
  \bibinfo{pages}{103480}.
\bibitem[{Chen and Skelton(2020)}]{chen2020general}
\bibinfo{author}{Chen, M.}, \bibinfo{author}{Skelton, R.E.},
  \bibinfo{year}{2020}.
\newblock \bibinfo{title}{A general approach to minimal mass tensegrity}.
\newblock \bibinfo{journal}{Composite Structures} , \bibinfo{pages}{112454}.
\bibitem[{Faroughi et~al.(2015)Faroughi, Khodaparast and
  Friswell}]{faroughi2015non}
\bibinfo{author}{Faroughi, S.}, \bibinfo{author}{Khodaparast, H.H.},
  \bibinfo{author}{Friswell, M.I.}, \bibinfo{year}{2015}.
\newblock \bibinfo{title}{Non-linear dynamic analysis of tensegrity structures
  using a co-rotational method}.
\newblock \bibinfo{journal}{International Journal of Non-Linear Mechanics}
  \bibinfo{volume}{69}, \bibinfo{pages}{55--65}.
\bibitem[{Fraddosio et~al.(2017)Fraddosio, Marzano, Pavone and
  Piccioni}]{fraddosio2017morphology}
\bibinfo{author}{Fraddosio, A.}, \bibinfo{author}{Marzano, S.},
  \bibinfo{author}{Pavone, G.}, \bibinfo{author}{Piccioni, M.D.},
  \bibinfo{year}{2017}.
\newblock \bibinfo{title}{Morphology and self-stress design of v-expander
  tensegrity cells}.
\newblock \bibinfo{journal}{Composites Part B: Engineering}
  \bibinfo{volume}{115}, \bibinfo{pages}{102--116}.
\bibitem[{Fuller(1982)}]{fuller1982synergetics}
\bibinfo{author}{Fuller, R.B.}, \bibinfo{year}{1982}.
\newblock \bibinfo{title}{Synergetics: explorations in the geometry of
  thinking}.
\newblock \bibinfo{publisher}{Estate of R. Buckminster Fuller}.
\bibitem[{Goyal et~al.(2019)Goyal, Chen, Majji and Skelton}]{goyal2019motes}
\bibinfo{author}{Goyal, R.}, \bibinfo{author}{Chen, M.},
  \bibinfo{author}{Majji, M.}, \bibinfo{author}{Skelton, R.},
  \bibinfo{year}{2019}.
\newblock \bibinfo{title}{Motes: Modeling of tensegrity structures}.
\newblock \bibinfo{journal}{Journal of Open Source Software}
  \bibinfo{volume}{4}, \bibinfo{pages}{1613}.
\bibitem[{Goyal et~al.(2020)Goyal, Chen, Majji and
  Skelton}]{goyal2020gyroscopic}
\bibinfo{author}{Goyal, R.}, \bibinfo{author}{Chen, M.},
  \bibinfo{author}{Majji, M.}, \bibinfo{author}{Skelton, R.},
  \bibinfo{year}{2020}.
\newblock \bibinfo{title}{Gyroscopic tensegrity robots}.
\newblock \bibinfo{journal}{IEEE Robotics and Automation Letters} .
\bibitem[{Goyal and Skelton(2019)}]{goyal2019tensegrity}
\bibinfo{author}{Goyal, R.}, \bibinfo{author}{Skelton, R.E.},
  \bibinfo{year}{2019}.
\newblock \bibinfo{title}{Tensegrity system dynamics with rigid bars and
  massive strings}.
\newblock \bibinfo{journal}{Multibody System Dynamics} \bibinfo{volume}{46},
  \bibinfo{pages}{203--228}.
\bibitem[{Guest(2011)}]{guest2011stiffness}
\bibinfo{author}{Guest, S.D.}, \bibinfo{year}{2011}.
\newblock \bibinfo{title}{The stiffness of tensegrity structures}.
\newblock \bibinfo{journal}{IMA Journal of Applied Mathematics}
  \bibinfo{volume}{76}, \bibinfo{pages}{57--66}.
\bibitem[{Kan et~al.(2017)Kan, Peng, Chen and Zhong}]{Kan2017A}
\bibinfo{author}{Kan, Z.}, \bibinfo{author}{Peng, H.}, \bibinfo{author}{Chen,
  B.}, \bibinfo{author}{Zhong, W.}, \bibinfo{year}{2017}.
\newblock \bibinfo{title}{A sliding cable element of multibody dynamics with
  application to nonlinear dynamic deployment analysis of clustered
  tensegrity}.
\newblock \bibinfo{journal}{International Journal of Solids \& Structures} ,
  \bibinfo{pages}{S0020768317304730}.
\bibitem[{Kan et~al.(2018)Kan, Peng, Chen and Zhong}]{kan2018nonlinear}
\bibinfo{author}{Kan, Z.}, \bibinfo{author}{Peng, H.}, \bibinfo{author}{Chen,
  B.}, \bibinfo{author}{Zhong, W.}, \bibinfo{year}{2018}.
\newblock \bibinfo{title}{Nonlinear dynamic and deployment analysis of
  clustered tensegrity structures using a positional formulation fem}.
\newblock \bibinfo{journal}{Composite Structures} \bibinfo{volume}{187},
  \bibinfo{pages}{241--258}.
\bibitem[{Kim et~al.(2020)Kim, Agogino and Agogino}]{kim2020rolling}
\bibinfo{author}{Kim, K.}, \bibinfo{author}{Agogino, A.K.},
  \bibinfo{author}{Agogino, A.M.}, \bibinfo{year}{2020}.
\newblock \bibinfo{title}{Rolling locomotion of cable-driven soft spherical
  tensegrity robots}.
\newblock \bibinfo{journal}{Soft Robotics} .
\bibitem[{Koohestani(2017)}]{koohestani2017analytical}
\bibinfo{author}{Koohestani, K.}, \bibinfo{year}{2017}.
\newblock \bibinfo{title}{On the analytical form-finding of tensegrities}.
\newblock \bibinfo{journal}{Composite Structures} \bibinfo{volume}{166},
  \bibinfo{pages}{114--119}.
\bibitem[{Lalvani(1996)}]{lalvani1996origins}
\bibinfo{author}{Lalvani, H.}, \bibinfo{year}{1996}.
\newblock \bibinfo{title}{Origins of tensegrity: views of emmerich, fuller and
  snelson}.
\newblock \bibinfo{journal}{International Journal of Space Structures}
  \bibinfo{volume}{11}, \bibinfo{pages}{27--27}.
\bibitem[{Lee and Lee(2016)}]{lee2016novel}
\bibinfo{author}{Lee, S.}, \bibinfo{author}{Lee, J.}, \bibinfo{year}{2016}.
\newblock \bibinfo{title}{A novel method for topology design of tensegrity
  structures}.
\newblock \bibinfo{journal}{Composite Structures} \bibinfo{volume}{152},
  \bibinfo{pages}{11--19}.
\bibitem[{Liu et~al.(2019)Liu, Zegard, Pratapa and Paulino}]{liu2019unraveling}
\bibinfo{author}{Liu, K.}, \bibinfo{author}{Zegard, T.},
  \bibinfo{author}{Pratapa, P.P.}, \bibinfo{author}{Paulino, G.H.},
  \bibinfo{year}{2019}.
\newblock \bibinfo{title}{Unraveling tensegrity tessellations for metamaterials
  with tunable stiffness and bandgaps}.
\newblock \bibinfo{journal}{Journal of the Mechanics and Physics of Solids} .
\bibitem[{Ma et~al.(2020)Ma, Chen and Skelton}]{ma2020design}
\bibinfo{author}{Ma, S.}, \bibinfo{author}{Chen, M.}, \bibinfo{author}{Skelton,
  R.E.}, \bibinfo{year}{2020}.
\newblock \bibinfo{title}{Design of a new tensegrity cantilever structure}.
\newblock \bibinfo{journal}{Composite Structures} , \bibinfo{pages}{112188}.
\bibitem[{Ma et~al.(2019a)Ma, Yuan and Samy}]{ma2019shape}
\bibinfo{author}{Ma, S.}, \bibinfo{author}{Yuan, X.F.}, \bibinfo{author}{Samy,
  A.}, \bibinfo{year}{2019}a.
\newblock \bibinfo{title}{Shape optimization of a new tensegrity torus}.
\newblock \bibinfo{journal}{Mechanics Research Communications}
  \bibinfo{volume}{100}, \bibinfo{pages}{103396}.
\bibitem[{Ma et~al.(2019b)Ma, Yuan and Xie}]{ma2019new}
\bibinfo{author}{Ma, S.}, \bibinfo{author}{Yuan, X.F.}, \bibinfo{author}{Xie,
  S.D.}, \bibinfo{year}{2019}b.
\newblock \bibinfo{title}{A new genetic algorithm-based topology optimization
  method of tensegrity tori}.
\newblock \bibinfo{journal}{KSCE Journal of Civil Engineering}
  \bibinfo{volume}{23}, \bibinfo{pages}{2136--2147}.
\bibitem[{Ma et~al.(2018)Ma, Zhang, Dobah, Scarpa, Fraternali, Skelton, Zhang
  and Hong}]{ma2018meta}
\bibinfo{author}{Ma, Y.}, \bibinfo{author}{Zhang, Q.}, \bibinfo{author}{Dobah,
  Y.}, \bibinfo{author}{Scarpa, F.}, \bibinfo{author}{Fraternali, F.},
  \bibinfo{author}{Skelton, R.E.}, \bibinfo{author}{Zhang, D.},
  \bibinfo{author}{Hong, J.}, \bibinfo{year}{2018}.
\newblock \bibinfo{title}{Meta-tensegrity: Design of a tensegrity prism with
  metal rubber}.
\newblock \bibinfo{journal}{Composite Structures} \bibinfo{volume}{206},
  \bibinfo{pages}{644--657}.
\bibitem[{Miranda et~al.(2020)Miranda, Singh, Santos, Fraternali
  et~al.}]{miranda2020mechanics}
\bibinfo{author}{Miranda, R.}, \bibinfo{author}{Singh, N.},
  \bibinfo{author}{Santos, F.}, \bibinfo{author}{Fraternali, F.}, et~al.,
  \bibinfo{year}{2020}.
\newblock \bibinfo{title}{Mechanics of smart origami sunscreens with energy
  harvesting ability}.
\newblock \bibinfo{journal}{Mechanics Research Communications} ,
  \bibinfo{pages}{103503}.
\bibitem[{Murakami(2001)}]{murakami2001static}
\bibinfo{author}{Murakami, H.}, \bibinfo{year}{2001}.
\newblock \bibinfo{title}{Static and dynamic analyses of tensegrity structures.
  part ii. quasi-static analysis}.
\newblock \bibinfo{journal}{International Journal of Solids \& Structures}
  \bibinfo{volume}{38}, \bibinfo{pages}{3615--3629}.
\bibitem[{Pajunen et~al.(2019)Pajunen, Johanns, Pal, Rimoli and
  Daraio}]{pajunen2019design}
\bibinfo{author}{Pajunen, K.}, \bibinfo{author}{Johanns, P.},
  \bibinfo{author}{Pal, R.K.}, \bibinfo{author}{Rimoli, J.J.},
  \bibinfo{author}{Daraio, C.}, \bibinfo{year}{2019}.
\newblock \bibinfo{title}{Design and impact response of 3d-printable
  tensegrity-inspired structures}.
\newblock \bibinfo{journal}{Materials \& Design} \bibinfo{volume}{182},
  \bibinfo{pages}{107966}.
\bibitem[{Rieffel et~al.(2009)Rieffel, Valero-Cuevas and
  Lipson}]{rieffel2009automated}
\bibinfo{author}{Rieffel, J.}, \bibinfo{author}{Valero-Cuevas, F.},
  \bibinfo{author}{Lipson, H.}, \bibinfo{year}{2009}.
\newblock \bibinfo{title}{Automated discovery and optimization of large
  irregular tensegrity structures}.
\newblock \bibinfo{journal}{Computers \& Structures} \bibinfo{volume}{87},
  \bibinfo{pages}{368--379}.
\bibitem[{Rimoli(2018)}]{rimoli2018reduced}
\bibinfo{author}{Rimoli, J.J.}, \bibinfo{year}{2018}.
\newblock \bibinfo{title}{A reduced-order model for the dynamic and
  post-buckling behavior of tensegrity structures}.
\newblock \bibinfo{journal}{Mechanics of Materials} \bibinfo{volume}{116},
  \bibinfo{pages}{146--157}.
\bibitem[{Skelton(2005)}]{skelton2005dynamics}
\bibinfo{author}{Skelton, R.}, \bibinfo{year}{2005}.
\newblock \bibinfo{title}{Dynamics and control of tensegrity systems}, in:
  \bibinfo{booktitle}{IUTAM symposium on vibration control of nonlinear
  mechanisms and structures}, \bibinfo{organization}{Springer}. pp.
  \bibinfo{pages}{309--318}.
\bibitem[{Skelton and de~Oliveira(2009)}]{skelton2009tensegrity}
\bibinfo{author}{Skelton, R.E.}, \bibinfo{author}{de~Oliveira, M.C.},
  \bibinfo{year}{2009}.
\newblock \bibinfo{title}{Tensegrity systems}. volume~\bibinfo{volume}{1}.
\newblock \bibinfo{publisher}{Springer}.
\bibitem[{Sultan et~al.(2002)Sultan, Corless and Skelton}]{sultan2002linear}
\bibinfo{author}{Sultan, C.}, \bibinfo{author}{Corless, M.},
  \bibinfo{author}{Skelton, R.E.}, \bibinfo{year}{2002}.
\newblock \bibinfo{title}{Linear dynamics of tensegrity structures}.
\newblock \bibinfo{journal}{Engineering Structures} \bibinfo{volume}{24},
  \bibinfo{pages}{671--685}.
\bibitem[{Xu et~al.(2018)Xu, Wang and Luo}]{xu2018improved}
\bibinfo{author}{Xu, X.}, \bibinfo{author}{Wang, Y.}, \bibinfo{author}{Luo,
  Y.}, \bibinfo{year}{2018}.
\newblock \bibinfo{title}{An improved multi-objective topology optimization
  approach for tensegrity structures}.
\newblock \bibinfo{journal}{Advances in Structural Engineering}
  \bibinfo{volume}{21}, \bibinfo{pages}{59--70}.
\bibitem[{Yang and Sultan(2019)}]{yang2019deployment}
\bibinfo{author}{Yang, S.}, \bibinfo{author}{Sultan, C.}, \bibinfo{year}{2019}.
\newblock \bibinfo{title}{Deployment of foldable tensegrity-membrane systems
  via transition between tensegrity configurations and tensegrity-membrane
  configurations}.
\newblock \bibinfo{journal}{International Journal of Solids and Structures}
  \bibinfo{volume}{160}, \bibinfo{pages}{103--119}.
\bibitem[{Yildiz and Lesieutre(2019)}]{yildiz2019effective}
\bibinfo{author}{Yildiz, K.}, \bibinfo{author}{Lesieutre, G.A.},
  \bibinfo{year}{2019}.
\newblock \bibinfo{title}{Effective beam stiffness properties of n-strut
  cylindrical tensegrity towers}.
\newblock \bibinfo{journal}{AIAA journal} \bibinfo{volume}{57},
  \bibinfo{pages}{2185--2194}.
\bibitem[{Yuan and Dong(2003)}]{yuan2003integral}
\bibinfo{author}{Yuan, X.}, \bibinfo{author}{Dong, S.}, \bibinfo{year}{2003}.
\newblock \bibinfo{title}{Integral feasible prestress of cable domes}.
\newblock \bibinfo{journal}{Computers \& structures} \bibinfo{volume}{81},
  \bibinfo{pages}{2111--2119}.
\bibitem[{Yuan et~al.(2017)Yuan, Ma and Jiang}]{yuan2017form}
\bibinfo{author}{Yuan, X.F.}, \bibinfo{author}{Ma, S.}, \bibinfo{author}{Jiang,
  S.H.}, \bibinfo{year}{2017}.
\newblock \bibinfo{title}{Form-finding of tensegrity structures based on the
  levenberg--marquardt method}.
\newblock \bibinfo{journal}{Computers \& Structures} \bibinfo{volume}{192},
  \bibinfo{pages}{171--180}.
\bibitem[{Zhang and Ohsaki(2006)}]{Zhang2006Adaptive}
\bibinfo{author}{Zhang, J.Y.}, \bibinfo{author}{Ohsaki, M.},
  \bibinfo{year}{2006}.
\newblock \bibinfo{title}{Adaptive force density method for form-finding
  problem of tensegrity structures}.
\newblock \bibinfo{journal}{International Journal of Solids \& Structures}
  \bibinfo{volume}{43}, \bibinfo{pages}{5658--5673}.
\bibitem[{Zhang et~al.(2018a)Zhang, Li, Zhu, Zhang and
  Xu}]{zhang2018automatically}
\bibinfo{author}{Zhang, L.Y.}, \bibinfo{author}{Li, S.X.},
  \bibinfo{author}{Zhu, S.X.}, \bibinfo{author}{Zhang, B.Y.},
  \bibinfo{author}{Xu, G.K.}, \bibinfo{year}{2018}a.
\newblock \bibinfo{title}{Automatically assembled large-scale tensegrities by
  truncated regular polyhedral and prismatic elementary cells}.
\newblock \bibinfo{journal}{Composite Structures} \bibinfo{volume}{184},
  \bibinfo{pages}{30--40}.
\bibitem[{Zhang and Feng(2017)}]{zhang2017initial}
\bibinfo{author}{Zhang, P.}, \bibinfo{author}{Feng, J.}, \bibinfo{year}{2017}.
\newblock \bibinfo{title}{Initial prestress design and optimization of
  tensegrity systems based on symmetry and stiffness}.
\newblock \bibinfo{journal}{International Journal of Solids and Structures}
  \bibinfo{volume}{106}, \bibinfo{pages}{68--90}.
\bibitem[{Zhang et~al.(2018b)Zhang, Zhang, Dobah, Scarpa, Fraternali and
  Skelton}]{zhang2018tensegrity}
\bibinfo{author}{Zhang, Q.}, \bibinfo{author}{Zhang, D.},
  \bibinfo{author}{Dobah, Y.}, \bibinfo{author}{Scarpa, F.},
  \bibinfo{author}{Fraternali, F.}, \bibinfo{author}{Skelton, R.E.},
  \bibinfo{year}{2018}b.
\newblock \bibinfo{title}{Tensegrity cell mechanical metamaterial with metal
  rubber}.
\newblock \bibinfo{journal}{Applied Physics Letters} \bibinfo{volume}{113},
  \bibinfo{pages}{031906}.

\end{thebibliography}

\end{document}